\documentclass[onefignum,onetabnum]{siamonline190516}

\usepackage{amsfonts}
\usepackage{graphicx}
\usepackage{epstopdf}
\usepackage{amsopn}
\usepackage{bm}
\usepackage{amsmath}
\usepackage{amssymb}
\usepackage{graphicx}
\usepackage{color}
\usepackage{url}
\usepackage{subcaption}
\renewcommand{\S}{{\mathcal S}}
\newcommand{\I}{\iota}
\newcommand{\R}{\mathbb{R}}
\newcommand{\Pro}{\mathbb{P}}
\newcommand{\Z}{\mathbb{Z}}
\newcommand{\C}{\mathbb{C}}
\newcommand{\Cm}{\mathcal{C}}
\newcommand{\Sm}{\mathcal{S}}
\newcommand{\B}{\mathcal{B}}

\newcommand{\T}{\mathcal{T}}
\newcommand{\sign}{\text{sign}}
\newcommand\K{\mathbb{K}}
\newcommand\MOD{\textrm{ mod }}

\DeclareMathOperator{\rank}{rank}

\newcommand{\rev}[1]{{\color{black}{#1}}}

\ifpdf
\DeclareGraphicsExtensions{.eps,.pdf,.png,.jpg}
\else
\DeclareGraphicsExtensions{.eps}
\fi

\usepackage{enumitem}
\setlist[enumerate]{leftmargin=.5in}
\setlist[itemize]{leftmargin=.5in}

\newsiamremark{remark}{Remark}
\newsiamremark{example}{Example}
\newsiamremark{hypothesis}{Hypothesis}
\crefname{hypothesis}{Hypothesis}{Hypotheses}
\newsiamthm{conj}{Conjecture}
\newsiamthm{prop}{Proposition}
\newsiamthm{cor}{Corollary}

\headers{Crystallographic phase retrieval}{T. Bendory and D. Edidin}

\title{Toward a mathematical theory of the crystallographic phase retrieval problem}

\author{Tamir Bendory\thanks{School of Electrical Engineering, Tel Aviv University, Tel Aviv, Israel
		(\email{bendory@tauex.tau.ac.il}).  Partially supported by the NSF-BSF award 2019752, and the Zimin Institute for Engineering Solutions Advancing Better Lives.}
	\and Dan Edidin\thanks{Department of Mathematics, University of Missouri, Columbia, MO, USA 
		(\email{edidind@missouri.edu}). Partially supported by Simons Collaboration Grant 315460.}}
\begin{document}

\maketitle

\begin{abstract}
 Motivated by the X-ray crystallography technology to  determine the atomic structure of biological molecules, we study the crystallographic phase retrieval problem,  arguably the leading and hardest  phase retrieval setup. This problem entails recovering a K-sparse signal of length $N$ from its Fourier magnitude or, equivalently, from its periodic auto-correlation. 
 Specifically, this work focuses on the fundamental question of uniqueness: what is the maximal sparsity level $K/N$ that allows unique mapping between a signal and its Fourier magnitude, up to intrinsic symmetries.
 We design a systemic computational technique to affirm uniqueness for any specific pair $(K,N)$, and establish the following conjecture: the Fourier magnitude determines a generic signal uniquely, up to intrinsic symmetries, as long as $K/N\leq 1/2$. 
 Based on group-theoretic considerations and an additional computational technique, we formulate a second conjecture: if $K/N < 1/2$, then
 for any signal the set of solutions to the crystallographic phase retrieval
 problem has measure zero in the set of all signals with a given Fourier magnitude.
  Together, these conjectures  constitute the first attempt to establish a mathematical theory for the crystallographic phase retrieval problem.
\end{abstract}



\section{Introduction}

\subsection{Problem formulation}
\label{sec:crystallographic_phase_retrieval}
The \emph{crystallographic phase retrieval problem} entails recovering a K-sparse signal $x_0\in\R^N$ from its Fourier magnitude
\begin{equation} \label{eq:model}
y_0 = |Fx_0|,
\end{equation}
where $F\in\C^{N\times N}$ is the discrete-time Fourier (DFT) matrix, and the absolute value is taken entry-wise. 
The problem can be equivalently formulated as recovering $x_0$ from its \emph{periodic
	auto-correlation}: 
\begin{equation} \label{eq:ac}
a_{x_0}[\ell] = \sum_{i=0}^{N-1}x_0[i]{x_0[(i+\ell)\MOD
N]},
\end{equation}
since $Fa_{x_0} = |Fx_0|^2$. 

A useful interpretation of the  crystallographic phase retrieval  problem  is  as  a feasibility problem of  finding the intersection of two non-convex sets 
\begin{equation} \label{eq:problem}
x_0\in\B \cap \Sm.
\end{equation}
Here, the set $\B$ describes all signals with the given Fourier magnitude
\begin{equation} \label{eq:B}
\B:= \left\{ x :  y_0 = |Fx|  \right\},
\end{equation}
or, equivalently, with the same periodic auto-correlation:
\begin{equation*}
\B:= \left\{ x : a_x[\ell] = a_{x_0}[\ell] \quad \text{for all} \quad  \ell=0,\ldots,N-1  \right\}.
\end{equation*}
The set  $\Sm$ consists of all signals with at most $K$ non-zero value entries, that is, all signals for which the support set
\begin{equation} \label{eq:S}
S=\{n: x[n]\neq 0\}\subseteq[0,N-1],
\end{equation}
obeys $|S|\leq K$.
Importantly,  since the Fourier magnitude  and the sparsity level of the signal remain unchanged under sign change, circular shift, and reflection, the signal can be recovered only up to these three \emph{intrinsic symmetries}. 
A rigorous definition of the group of intrinsic symmetries|occasionally referred to as trivial ambiguities in the phase retrieval literature|is provided in Section~\ref{sec:intrinsic_symmetries}.

%

The main objective of this paper is to characterize the sparsity level $K/N$ that allows unique mapping, up to  intrinsic symmetries, between a signal and its periodic auto-correlation. In other words: the sparsity level under which the periodic auto-correlation mapping $x\mapsto a_x$ is injective.
For general signals, it is not difficult to bound $K/N$ from above. 
To this end, we note that the periodic auto-correlation is invariant under 
reflection, namely, $a_x[i]= {a_x[N-i]}$.
The set of vectors with support set \rev{contained in} $S$ is a $K$-dimensional linear
subspace~$L_S$ and thus the periodic auto-correlation is a quadratic function $L_S \to
\R^{\lfloor N/2\rfloor+1}$.  
\rev{Therefore}, by counting dimensions, we do not expect to be able to obtain unique recovery unless $K \leq \lfloor N/2 \rfloor+1$, even if the support $S$ is known.
This simple argument establishes a necessary condition on $K$; Conjecture~\ref{conj.main.descriptive}, formulated in Section~\ref{sec:contribution}, states that this is also a sufficient condition for uniqueness if the signal is generic. 

\subsection{X-ray crystallography} \label{sec:crystallography}

This work is motivated by X-ray crystallography|a prevalent technology for determining the
3-D atomic structure of  molecules~\cite{millane1990phase}:  nearly 50,000 new crystal structures are added each year to the Cambridge Structural Database,  the world's repository for crystal structures~\cite{CSD}.
While the crystallographic problem is the leading (and arguably the hardest) phase retrieval problem, its mathematical characterizations have not been analyzed thoroughly so far.

The mathematical model of X-ray crystallography is introduced
and discussed at length in~\cite{elser2018benchmark}.
For completeness, we provide a concise summary. 
In X-ray crystallography, the signal is the electron density function of the
crystal|a periodic arrangement of a repeating, compactly supported
 unit
\begin{equation} \label{eq:repeated_motif}
x_c(t) = \sum_{s\in S} x(t-s),
\end{equation}
where ${x}$ is the repeated motif
and $S$ is a large, but finite subset of a lattice $\Lambda\subset\R^D$; the dimension $D$ is usually two or three.
The crystal is illuminated with a beam of X-rays producing a diffraction pattern, which is  equivalent to the   magnitude of the Fourier transform of the  crystal:
\begin{equation}
\begin{split}
\left\vert \hat{x}_c(k) \right\vert^2 &= \left\vert \int_{\R^D} x_c(t)e^{-\I\langle t,k \rangle} dt \right\vert^2\\
&= \left\vert \int_{\R^D} \sum_{s\in  S} x(t-s)e^{-\I\langle t,k \rangle} dt \right\vert^2\\
&=  \left\vert \sum_{s\in S}e^{-\I\langle s,k \rangle} \int_{\R^D} x(t)e^{-\I\langle t,k \rangle} dt\right\vert^2\\
&=  \left\vert \hat{s}(k)\right\vert^2\left\vert \hat{x}(k)\right\vert^2,
\end{split}
\end{equation}
where \rev{$\hat{x}$ and $\hat{s}$ are, respectively, the Fourier transforms of the signal $x$ and a Dirac ensemble defined on $S$}.
As the size of the set $S$ grows (the size of the crystal), the support of the function~$\hat{s}$ is more concentrated in the dual lattice $\Lambda^*$ \footnote{The dual $\Lambda^*$ of a lattice $\Lambda\subset \R^D$ is the  lattice of all vectors $x\in\text{span}(\Lambda) \subset \R^D$ such that $\langle x,y\rangle$ is an integer for all $y\in\Lambda$. For example, if $\Lambda = 2 \Z \subset \R$ then $\Lambda^* = {1\over 2}\Z \subset \R$.}. 
Thus, the diffraction pattern is approximately equal to a discrete set of samples of $|\hat{x}|^2$ on $\Lambda^*$, called \emph{Bragg peaks}.
This implies that the acquired data is
the Fourier magnitude of a $\Lambda$-periodic signal on $\R^D$
(or equivalently a signal on $\R^D/\Lambda$), defined by its {Fourier series}:
\begin{equation}
x(t) = \frac{1}{\text{Vol}(\Lambda)}\sum_{k\in\Lambda^*}\hat{x}(k) e^{\I\langle k,t\rangle}.
\end{equation}
This signal is supported  only at the sparsely-spread positions of atoms.
Elser estimated the typical number of strong scatters in a protein crystal (e.g., nitrogen, carbon, oxygen atoms)  to be $K/N\sim 0.01$~\cite{elser2017complexity}. 
In practice, the  data also follows a Poisson distribution (namely, noise), whose mean is the signal. 

The gaps between the idealized mathematical model~\eqref{eq:model} the phase retrieval crystallographic problem as it appears in X-ray experiments are discussed in Section~\ref{sec:contribution}.

\subsection{Notation}
Throughout the work, all indices should be considered as modulo $N$. For instance, $x[-i]= x[N-i]$.
The Fourier transform and the conjugate of a signal $x$ are denoted, respectively, by $\hat{x}$ (namely, $Fx=\hat{x}$) and $\bar{x}$.
An entry-wise product between two vectors $u$ and $v$ is denoted by $u\odot v$ so that $(u\odot v)[n] = u[n]v[n]$; absolute value of a vector $|u|$ refers to an entry-wise operation, that is, $|u|[n] = |u[n]|$.
For a set $S\subseteq [0,N-1]$, we let $L_S$ be the subspace \rev{of signals with support contained in $S$}~\eqref{eq:S}, and  denote its cardinality by $|S|$
or $K$.
While most of this work is focused on real signals, some of the results hold for complex signals as well. We use the notation $\K^N$ to denote either vector space
$\R^N$ or $\C^N$, and define the periodic auto-correlation by $$a_{x}[\ell] = \sum_{i=0}^{N-1}x[i]\overline{x[(i+\ell)\MOD N]}.$$ It satisfies the conjugation-reflection symmetry $a_x[\ell]=\overline{a[N-\ell]}$.

To ease notation, we make two assumptions that do not affect the
generality of the results.  First, we consider 1-D signals; the
extension to a high dimensional setting is straight-forward, and is
discussed in Section~\ref{sec:higher_dimensions}. Second, hereafter we
assume that $N$ is even; all the results hold for odd $N$, where the
only change is that $N/2$ should be replaced by $\lfloor N/2\rfloor$.

\section{Contribution and perspective}
\label{sec:contribution}

To our best knowledge, this is the first work to rigorously study the mathematics of the crystallographic phase retrieval problem.  
While general uniqueness results are currently beyond reach, the main contribution of this paper is \rev{conjecturing that the Fourier magnitude determines uniquely, up to intrinsic symmetries,  almost all K-sparse signals as long as $K\leq N/2$. 
	This number is significantly larger than the typical number of strong scatters in a protein crystal  which was estimated to be $K/N\sim 0.01$~\cite{elser2017complexity}.
	In this sense, our conjecture suggests that (under
	the stated conditions) crystallographers should not worry too much
	about uniqueness: the data (i.e., Fourier magnitude) usually determines the sought signal (e.g., the atomic structure of a molecule) uniquely.
	 More formally, the main conjecture of this paper states the following. } 
\begin{conj} \label{conj.main.descriptive}
	Suppose that $x$ is a $K$-sparse generic signal with $K\leq N/2$, whose periodic  auto-correlation $a_x$ has more than $K$ non-zero entries. Then,  $a_x=a_{x'}$ implies that $x'$ is obtained from $x$ through an intrinsic symmetry. In other words, under the stated conditions,  the periodic  auto-correlation mapping $x\mapsto a_x$ is injective, up to intrinsic symmetries, for almost all signals.
\end{conj}
In Section~\ref{sec:theory}, we state the conjecture more precisely
and  establish a systematic computational technique to verify it for any particular \rev{pair} $(K,N)$. 

Conjecture~\ref{conj.main.descriptive} 
 {puts a structural requirement}
on the signal's support $S$: the cardinality of the periodic auto-correlation's support should be larger than $K$. 
However, this condition seems to constitute only a minor restriction as it   
is almost always met.  
To comprehend the last statement, we need the notion of {\em cyclic difference set}, denoted by
$S-S$, which includes all the differences of a 
set~$S$, that is, $\{ j-i| i, j \in S\}$. 
In the set $S-S$ we consider only the first $N/2+1$ entries because of the reflection symmetry; see a formal definition in Section~\ref{sec.difference}.
For example, if $S = \{0,1,2,5\} \subset [0,8]$, then
$S-S =\{0,1,2,3,4\}$.
The notion of cyclic difference set is useful since it defines the support of the periodic auto-correlation~\eqref{eq:ac}.
Specifically,  Conjecture~\ref{conj.main.descriptive} assumes $|S-S|>K$.
Generally, proving tight bounds on the probability to obtain $|S-S|>K$  (as a function of $N$ and $K$) is a very challenging combinatorial problem.
Nevertheless, empirical examination  is easy:
Figure~\ref{fig:hist} shows the empirical distribution of $|S-S|$ for different values of $K$. As can be seen, in \emph{all trials} we \rev{obtained} $|S-S|>K$, as desired, even for rather small value of $K=5$. In Proposition~\ref{pro:prime} we also prove that if $N$ is a prime number, then $|S-S|>K$ with probability one as $N\to\infty$; see a detailed discussion in  Appendix~\ref{app.difference}.
\rev{The empirical affirmation of the condition $|S-S|>K$ implies that we expect Conjecture~\ref{conj.main.descriptive} to hold for almost all $K$-sparse signal provided that $K/N\leq 1/2$.} 

\begin{figure}[ht]
	\begin{subfigure}[h]{0.5\columnwidth}
		\centering
		\includegraphics[width=\columnwidth]{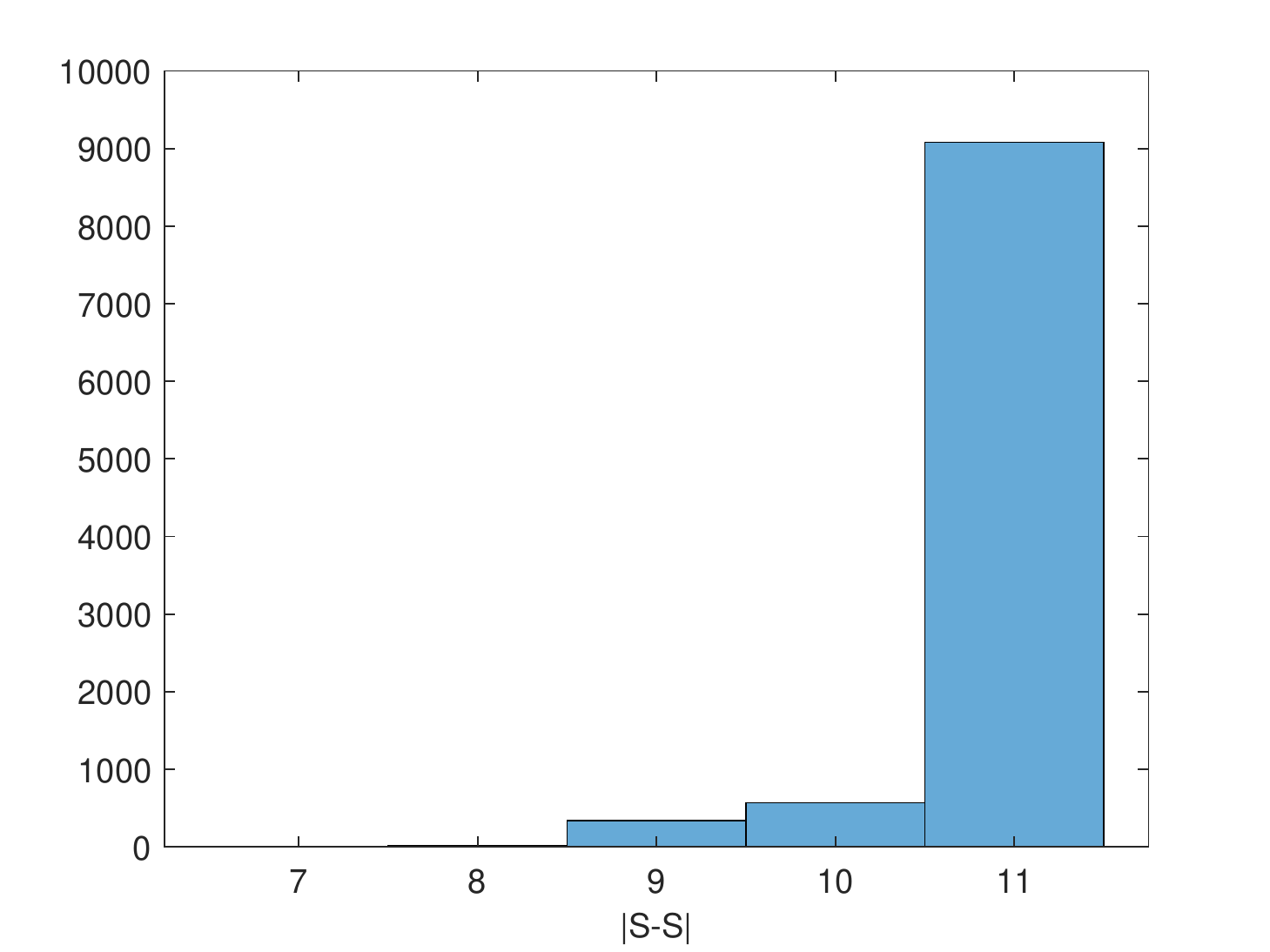}
		\caption{K=5}
	\end{subfigure}
	\hfill
	\begin{subfigure}[h]{0.5\columnwidth}
		\centering
		\includegraphics[width=\columnwidth]{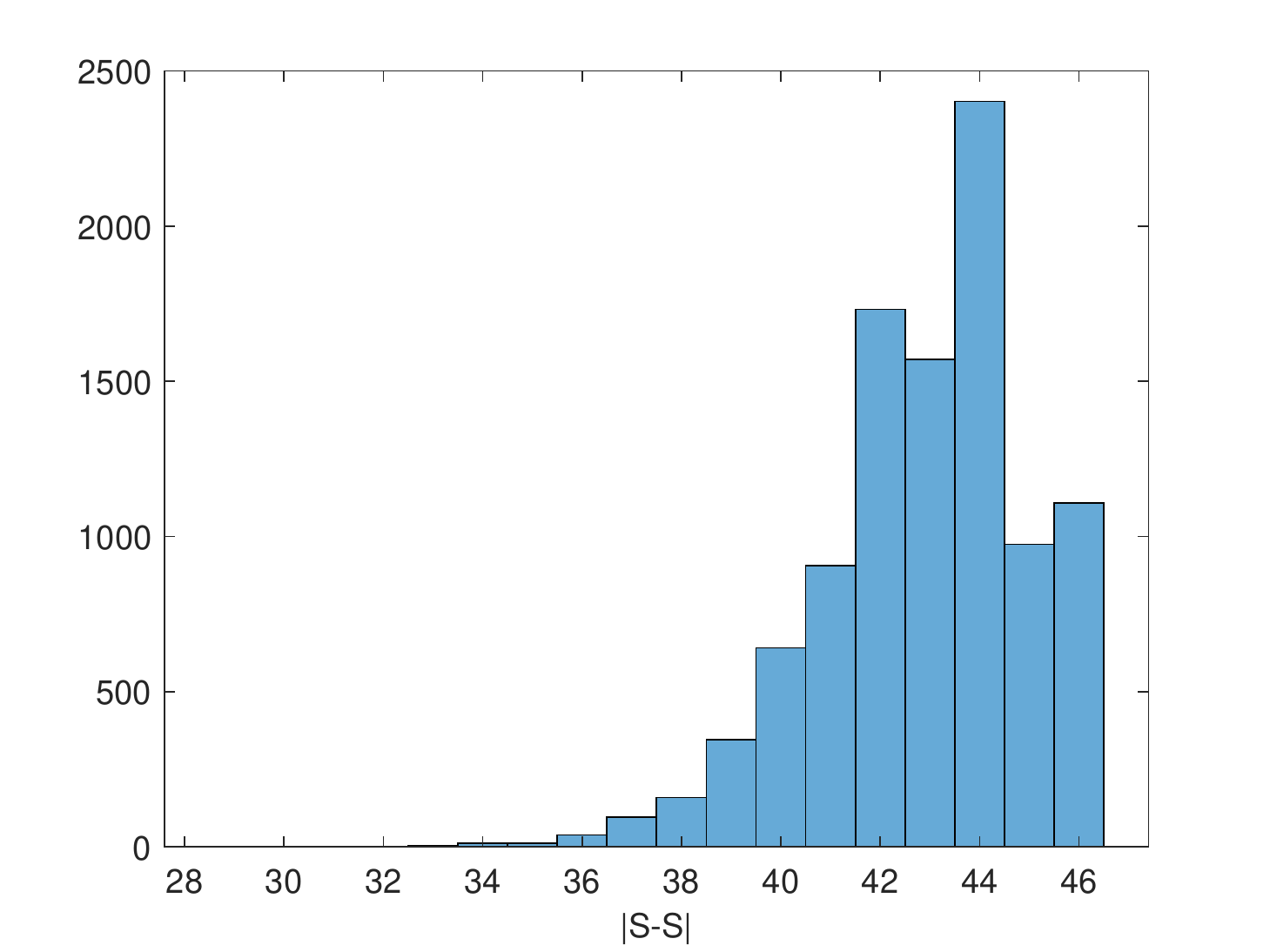}
		\caption{K=10}
	\end{subfigure}
	
	\begin{subfigure}[h]{0.5\columnwidth}
		\centering
		\includegraphics[width=\columnwidth]{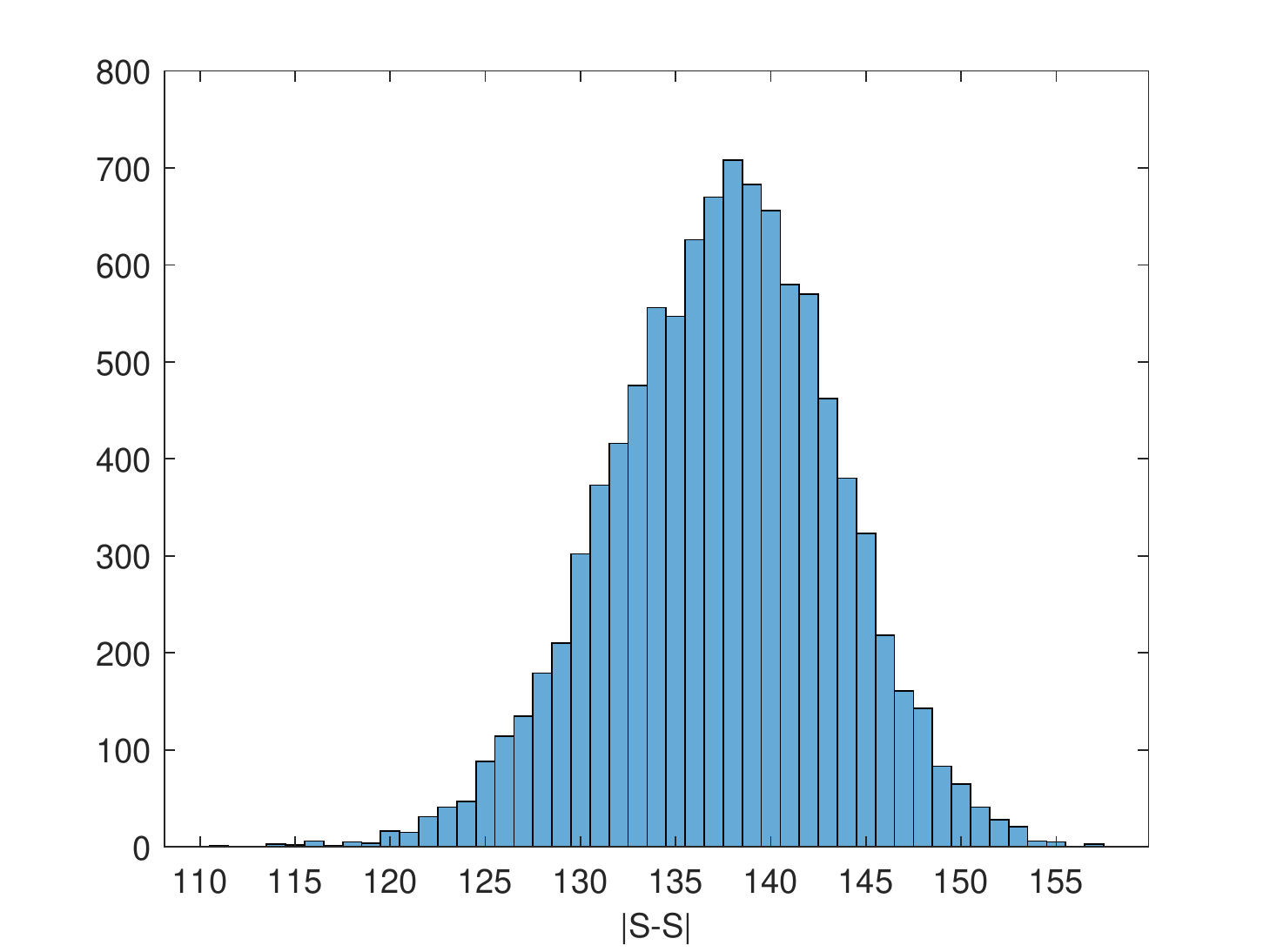}
		\caption{K=20}
	\end{subfigure}
	\hfill
	\begin{subfigure}[h]{0.5\columnwidth}
		\centering
		\includegraphics[width=\columnwidth]{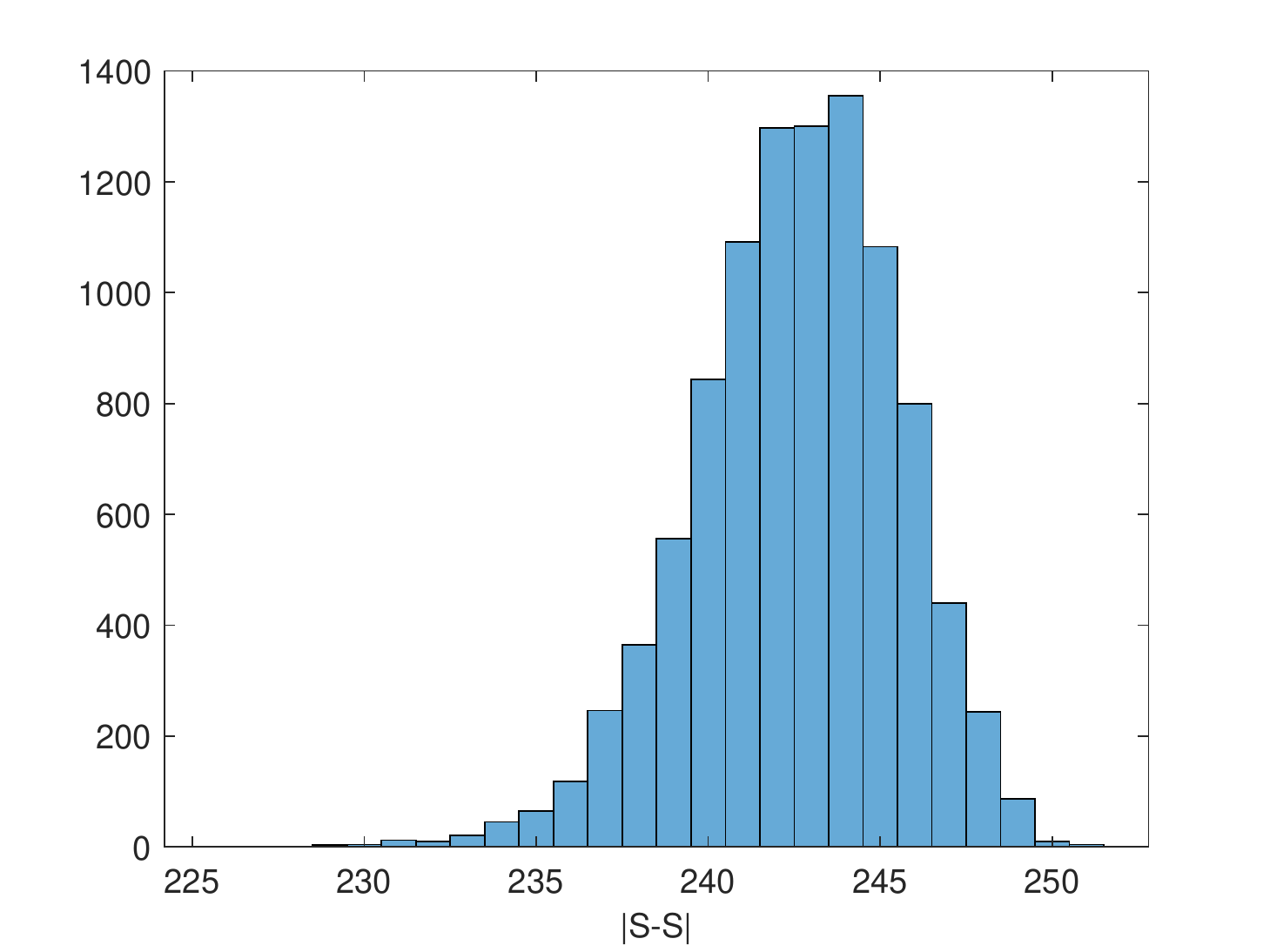}
		\caption{K=40}
	\end{subfigure}
	
	\caption{\label{fig:hist} 
	  The empirical distributions (histograms) of the  auto-correlation's support $|S-S|$  for $N=500$ and randomly sampled subsets of size $|S|=K=\{5,10,20,40\}$. Each histogram is composed of $10^4$ trials. 
	In all trials, we obtained $|S-S|>K$, as Conjecture~\ref{conj.main.descriptive} requires. 	
} 
\end{figure}

Our second conjecture states that even if there exist additional solutions (i.e., lack of uniqueness), the set of all solutions is of measure zero.
\rev{Therefore, in the worst case, there are only a few K-sparse signals that agree with the observed Fourier magnitude.}
 Importantly, the conjecture applies to \emph{all signals} and does not impose any structural condition on the support.  
\begin{conj} \label{conj.weak.descriptive}
Suppose that $x$ is a $K$-sparse signal 
with $K/N < 1/2$. Then, the set of $K$-sparse signals  with periodic auto-correlation $a_x$ is of measure zero.
 \end{conj}
Based on group-theoretic considerations, Section~\ref{sec:group_theory} introduces Conjecture~\ref{conj.weak.descriptive} 
in technical terms, and develops a computational confirmation technique for any particular pair of $(K,N)$.
In Section~\ref{sec:higher_dimensions} we discuss the extension of Conjectures~\ref{conj.main.descriptive} and~\ref{conj.weak.descriptive} to higher dimensions, which is straight-forward.

Before moving on to surveying related literature, we wish to list the gaps between the model considered in this paper~\eqref{eq:model} and the crystallographic phase retrieval problem in practice.
A full mathematical theory of the crystallographic phase retrieval problem should account for the following aspects, which are beyond the scope of this work.

\begin{itemize}
	\item \textbf{Rigorous uniqueness results}: This work
          formulates conjectures and provides computational means to
          check unique mapping between a generic signal and its
          Fourier magnitude for any particular pair $(N,K)$. A
          complete theory should provide a rigorously proven bound on
          $K$ (as a function of $N$) that allows unique mapping for a
          certain class of signals.
	
	\item \textbf{Class of signals:} This work puts a special focus on generic signals (e.g., the non-zeros entries are drawn from a continuous distribution) and also discusses binary signals (i.e., the non-zeros entries are all ones). In practice, however, the model should account for sparse signals whose non-zero entries are taken from a finite (small) alphabet; this alphabet models the relevant type of atoms, such as hydrogen, oxygen, carbon, nitrogen, and so on. This model is more involved and requires intricate combinatorial calculations.
	
	\item \textbf{Noisy data}: In an X-ray crystallography experiment, the data is contaminated with noise, which is characterized by Poisson statistics. In this case, the intersection $\B\cap\Sm$ is empty, and the goal is to find a point close (in some metric) to the intersection. 
	
	\item  \textbf{Provable algorithms:} As discussed in Section~\ref{sec:numerics}, the state-of-the-art algorithms for phase retrieval are based on (non-convex) variations of the Douglas-Rachford splitting method. While Douglas-Rachford is fairly well understood for convex setups, the analysis of its non-convex analogues for phase retrieval is lacking. We refer the readers for several recent works on the topic~\cite{hesse2013nonconvex,li2016douglas,phan2016linear,elser2017complexity,lindstrom2018survey,levin2019note}, and to Appendix~\ref{sec:numerics} for a discussion on the computational complexity of the problem. 
	
	\item \textbf{Sampling:} In this work, we consider a discrete setup. In practice, however, the signal is continuous and its Fourier magnitude is measured on a Cartesian grid. Thus, sampling effects should be taken into account. 
	
	\item \textbf{Additional information:} In many setups, the scientist possess some additional information about the underlying signal; this information may significantly alleviate the reconstruction process. For example, some X-ray crystallography algorithms incorporate  knowledge of the minimum atom-atom
	distance, the presence of a known number of heavy atoms, or even the expected histogram of the signal values~\cite{elser2018benchmark}.
	Such information may allow recovering a non-sparse signal even in the regime $K>N/2$.
	We note that a similar analysis has been conducted for the problem of retrieving a 1-D signal from its aperiodic auto-correlation~\cite{beinert2018enforcing}. 
	
\end{itemize}



\section{Prior art}
\label{sec:literature}

As far as we know, the first  to study  an instance of the crystallographic phase retrieval problem (from the mathematical and algorithmic perspectives) was Elser in~\cite{elser2017complexity}. This paper  discusses the hardness of the crystallographic phase retrieval problem for binary signals  and linked it to other domains of research, such as cryptography.
The subject was further investigated in~\cite{wong2019phase}. In particular, it was shown that the solution of the ``box relaxation''  optimization problem 
\begin{equation} \label{eq:box}
\text{find } x\in[0,1]^N \quad \text{ subject to } \quad  |Fx|=y_0,
\end{equation}
is the underlying binary signal $x_0\in\{0,1\}^N$.
In other words, under the measurement constraint, the solution of~\eqref{eq:box} cannot lie within the box $[0,1]^N$ but only on the vertices. 
In addition, uniqueness results were derived for the cases of $K=1,2,3,N-3,N-2,N-1$. 
The general crystallographic phase retrieval problem~\eqref{eq:problem} has not been previously studied.

The crystallographic phase retrieval problem is an instance of  a broader
class problems. The (noiseless) phase retrieval problem
is any problem of the form 
\begin{equation} \label{eq:phase_retrieval}
\text{find } x\in\Cm \quad \text{subject to} \quad |Ax|=y_0,
\end{equation}
where $y_0$ is the observation, $A$ is a Fourier-type matrix (e.g., DFT, oversampled DFT, short-time Fourier transform) and the set $\Cm$ corresponds
to the constraints dictated by the particular application. 
For example, in the crystallographic phase retrieval problem, $A$ is the DFT matrix, and $\Cm$ is the set of all K-sparse signals, denoted by
$\Sm$ in~\eqref{eq:model}.

An important example of a phase retrieval problem arises
in coherent diffraction imaging. Here, an object is illuminated with a
coherent wave and the  diffraction intensity pattern (equivalent to the Fourier magnitude of the signal) is measured. As an additional constraint, usually  the support of the signal is assumed to be known (i.e.,\ the signal is known  to be zero outside of some region)~\cite{shechtman2015phase,Bendory2017}. This condition is equivalent to requiring that the signal lies in the column space of an over-sampled DFT matrix. 
If the over-sampling ratio 
is at least two (namely, the number of rows is at least twice the number of columns), 
the problem is equivalent to recovering a signal from its \emph{aperiodic auto-correlation}:
\begin{equation} \label{eq:non-periodic_ac}
b_x[\ell] = \sum_{i=0}^{N-\ell-1}x[i]\overline{x[i+\ell]}, \qquad \ell=0,\ldots,N-1.
\end{equation}

Generally, it is known that there are $2^{N-2}$
non-equivalent
1-D signals that are mapped to the same  aperiodic auto-correlation  (rather than infinitely many signals that are mapped to the same periodic auto-correlation~\eqref{eq:ac}), and the  geometry of the problem has been investigated meticulously~\cite{beinert2015ambiguities,edidin2019geometry}.
It is further known that the  number of solutions can be reduced when additional information is available~\cite{beinert2018enforcing,huang2016phase}.
In more than one dimension, almost any signal can be determined uniquely from its aperiodic auto-correlation~\cite{hayes1982reconstruction}.
Nevertheless, in practice it might be notoriously difficult to recover the signal due to  severe conditioning issues~\cite{barnett2018geometry}.
When the signal is sparse, the  recovery problem is significantly easier~\cite{ranieri2013phase}.
In particular, a polynomial-time algorithms was devised to provably recover almost all signals when $K=O(N^{1/2-\varepsilon})$ under some  constraints on the distribution of the support entries~\cite{jaganathan2017sparse}.
 

Another noteworthy phase retrieval application is  ptychography. Here, 
a moving probe is used to sense multiple diffraction
measurements~\cite{rodenburg2008ptychography,
  maiden2011superresolution}. If the shape of the probe is known
precisely\footnote{In practice, the probe shape is unknown precisely,
  and thus the goal is to recover the signal and the shape of the
  probe simultaneously; for a theoretical analysis,
  see~\cite{bendory2019blind}.}, then the problem is equivalent to
measuring the short-time Fourier transform (STFT) magnitude of the
signal, so that the matrix $A$ in~\eqref{eq:phase_retrieval}
represents an STFT
matrix~\cite{marchesini2016alternating,jaganathan2016stft,bendory2017non,iwen2016fast,pfander2019robust}. Additional
settings that were analyzed mathematically include
holography~\cite{gabor1948new,barmherzig2019holographic}, vectorial
phase retrieval~\cite{raz2013vectorial}, and ultra-short laser
characterization~\cite{trebino2012frequency,bendory2017uniqueness,bendory2018signal},

In addition to the aforementioned phase retrieval setups, we
mention a distinct line of work which studies a {toy model}
where |to facilitate the mathematical and algorithmic analysis|the
Fourier-type matrix in~\eqref{eq:phase_retrieval} is replaced by a
``sensing matrix'' $A\in\C^{M\times N}$.  In particular, many
papers consider the case where the entries of $A$ are drawn
i.i.d.\ from a normal distribution with $M\geq 2N$.  For instance, it was shown that for
generic $A$ only $M = 2N-1$ (for real) and $M=4N-4$ (for complex) observations are required to characterize all
signals uniquely~\cite{balan2006signal,conca2015algebraic}. Moreover,
based on convex and non-convex optimization techniques, provable
efficient algorithms were devised that estimate the signals stably
with merely $M=O(N)$ observations;
see~\cite{candes2015phase,waldspurger2015phase,candes2015phase,chen2017solving,sun2018geometric,goldstein2018phasemax}
to name a few.  Later on, the analysis was extended to more intricate
models, such as randomized Fourier
matrices~\cite{candes2015phase,gross2017improved}. In addition, some
papers considered similar randomized setups, when $N<M$ and the signal
is
sparse~\cite{cai2016optimal,ohlsson2014conditions,soltanolkotabi2019structured}.
This research thread led to new theoretical, statistical, and
computational results in a variety of fields, such as algebraic
geometry, statistics, and convex and non-convex optimization.
Nevertheless, its contribution to the \rev{crystallographic phase retrieval problem is
disputable: none
of the algorithms that were developed for randomized
sensing matrices  have been successfully implemented to X-ray crystallography~\cite{elser2018benchmark}.}
  In contrast, the algorithms that are used
routinely by practitioners are based on variations of the
Douglas-Rachford splitting scheme. The behavior of these algorithms
differs significantly from optimization-based algorithms and is far
from being understood; see an elaborated discussion in
Section~\ref{sec:numerics}.


\section{Uniqueness for generic signals}
\label{sec:theory}

In this section, we introduce our main conjecture on the uniqueness of generic signals and describe a set of computational tests to
verify it for any specific pair of $(N,K)$.

\subsection{Preliminaries}

We begin by formally introducing the intrinsic symmetries (i.e., trivial ambiguities) of the crystallographic phase retrieval problem, and 
discussing difference sets, multi-sets, and their connection with the uniqueness of binary signals.  

\subsubsection{Intrinsic symmetries and orbit recovery}
\label{sec:intrinsic_symmetries}
Unique mapping between a K-sparse signal  and its Fourier magnitude is possible
only up to three types of symmetries: circular shift, reflection through the
origin, and
global phase
change. These symmetries are frequently referred to
as trivial ambiguities in the phase retrieval literature~\cite{shechtman2015phase,Bendory2017}.
\begin{prop} \label{prop:trivial_symmetries}
	Let $x\in \K^N$ (either $\K=\C$ or $\K=\R$) be a K-sparse signal. Then, the following  are also K-sparse signals with the same Fourier magnitude:
	\begin{itemize}
		\item the signal $xe^{\I\phi}$ for some $\phi\in\R$ (if $\K=\R$, then it reduces to $\pm x$);
		\item the rotated signal $x^\ell[n] := {x[(n-\ell) \,\mathrm{ mod }\, N]}$ for some $\ell\in \Z$;
		\item the conjugate-reflected signal $\tilde{x}$, obeying $\tilde{x}[n] = \overline{x[-n\MOD N]}$.
	\end{itemize}
\end{prop}

These three types of symmetries form a symmetry group which we call
the {\em group of intrinsic symmetries}.
For $\K=\C$, a signal is invariant under the action of the group $D=
(S^1\times \Z_N)\ltimes \Z_2$, where $\ltimes$ denotes a semi-direct product.
The first $S^1$ corresponds to the phase  symmetry, $\Z_N$ corresponds to the group of $N$ cyclic shifts, and the last $\Z_2$ corresponds to the reflection symmetry; the last two symmetries generate the dihedral  group $D_{2N}$ of
symmetries of the regular $N$-gon. 
For $\K=\R$, the phase symmetry is replaced by a sign ambiguity $\Z_2=\pm 1$,
and the group of intrinsic symmetries reduces to $D= \Z_2 \times D_{2N}$.
Interestingly, an analog intrinsic symmetry group is formed when the crystallographic  phase retrieval problem is generalized to  any abelian finite group; see Appendix~\ref{sec:general_group}. 

Proposition~\ref{prop:trivial_symmetries} implies that the intersection $\B\cap\Sm$ is
invariant under the action of the group of intrinsic symmetries $D$.  In
particular, if $x\in\B\cap\Sm$, then so is $g\cdot x$ for
any element~$g$ in~$D$. In group theory terminology, the set of signals
$\{g\cdot x: g\in D\}$ is called the \emph{orbit of $x$ under
	$D$}. Therefore, our goal in this work to identify the regime in which
the intersection of $\Sm$ and $\B$ consists of a single orbit.  This
interpretation builds a connection between the crystallographic
phase retrieval problem (as well as other phase retrieval problems)
and other classes of \emph{orbit recovery problems}, such as
single-particle reconstruction using cryo-electron microscopy, and
multi-reference alignment; see for
instance~\cite{bandeira2017estimation,bendory2019single}.  Throughout this
paper we say that two
signals are {\em equivalent} if they lie in the same orbit under $D$. Otherwise, we say that the signals are non-equivalent.

We note that the two groups $S^1$ ($\Z_2$ for $\K=\R$) and the
dihedral group $D_{2N}$ play a different role in the analysis. The
dihedral group acts on the set $S$|the support of the signal|by permuting its
indices (recall that it is a subgroup of the permutation group). 
In particular, we say that $S$ and $S'$
are equivalent if $g\cdot S = S'$ for some element $g \in D_{2N}$.
The phase (or sign in the real case) symmetry affects only the values of the non-zero entries, and thus plays a lesser role for generic signals.

\subsubsection{Difference sets, multi-sets, collisions, and uniqueness for binary signals}
\label{sec.difference}
The support recovery analysis  is tightly related to the notion of cyclic difference sets.
Let us identify $[0,N-1]$ with the group $\Z_N$.  Then, 
there is an action of the group $\Z_2 = \pm 1$ on $\Z_N$ by
$n \mapsto -n \MOD N$; this action corresponds to the reflection symmetry of the periodic auto-correlation. The set of orbits under this action can be identified with the set $[0,N/2]$.
Given a subset $S \subset [0, N-1]$, we define the
cyclic difference set $S-S \subset [0, N/2]$
as the set of equivalence classes of  $\{ j-i \MOD \pm 1| i, j \in S\}$.
For example, if $S = \{0,1,2,5\} \subset [0,8]$, then
$S-S =\{0,1,2,3,4\}$.


We may also view $S-S$ as a multi-set, where we count the multiplicities of the differences. For the example above,  $S-S = \{0^4,1^2, 2^1,3^1, 4^2\}$. 
The cardinality of $|S-S|$ as a multi-set equals ${K+1\choose{2}}$
and thus depends only on $K$, but 
the cardinality of $|S-S|$ as a set depends on the particular subset.
Note that $S-S$ (either as a set or as a multi-set) is invariant under the action
of $D_{2N}$ on the set of subsets. Thus, equivalent subsets have the
same difference set. 

Multiplicities greater than one are occasionally referred to as \emph{collisions}; a collision-free subset is a subset whose corresponding multi-set has no multiplicity larger than one.
From phase retrieval standpoint, collisions are challenging since it is difficult to determine a priori how many pairs of support's entries are mapped into one auto-correlation entry. 
Unfortunately, the following
proposition shows that for any fixed value of $K/N$, collision-free sets do not exist if $N$ is sufficiently large.
\begin{prop}  \label{prop:collisions}
	For any $R \in (0,1]$, for $N$ sufficiently large (as a function of $R$) there does not exist
	a collision-free subset of size $K \geq R N$.
\end{prop}
Proposition~\ref{prop:collisions} is proven in Appendix~\ref{sec:proof_prop_collisions}.  \rev{Note that the proposition holds true even for an arbitrarily low sparsity level $K/N(R)$.}
Figure~\ref{fig:collision} shows empirically \rev{that} for a fixed $N$, collision-free subsets are rare unless $K$ is very small compared to $N$. 
Figure~\ref{fig:collision_fixed_ratio} exemplifies that if we keep the ratio $K/N$ fixed, in this case $K/N=0.01$|the expected density in proteins|
collision free subsets
are uncommon as $N$ grows. 

\begin{figure}[ht]
	\begin{subfigure}[h]{0.5\columnwidth}
		\centering
		\includegraphics[width=\columnwidth]{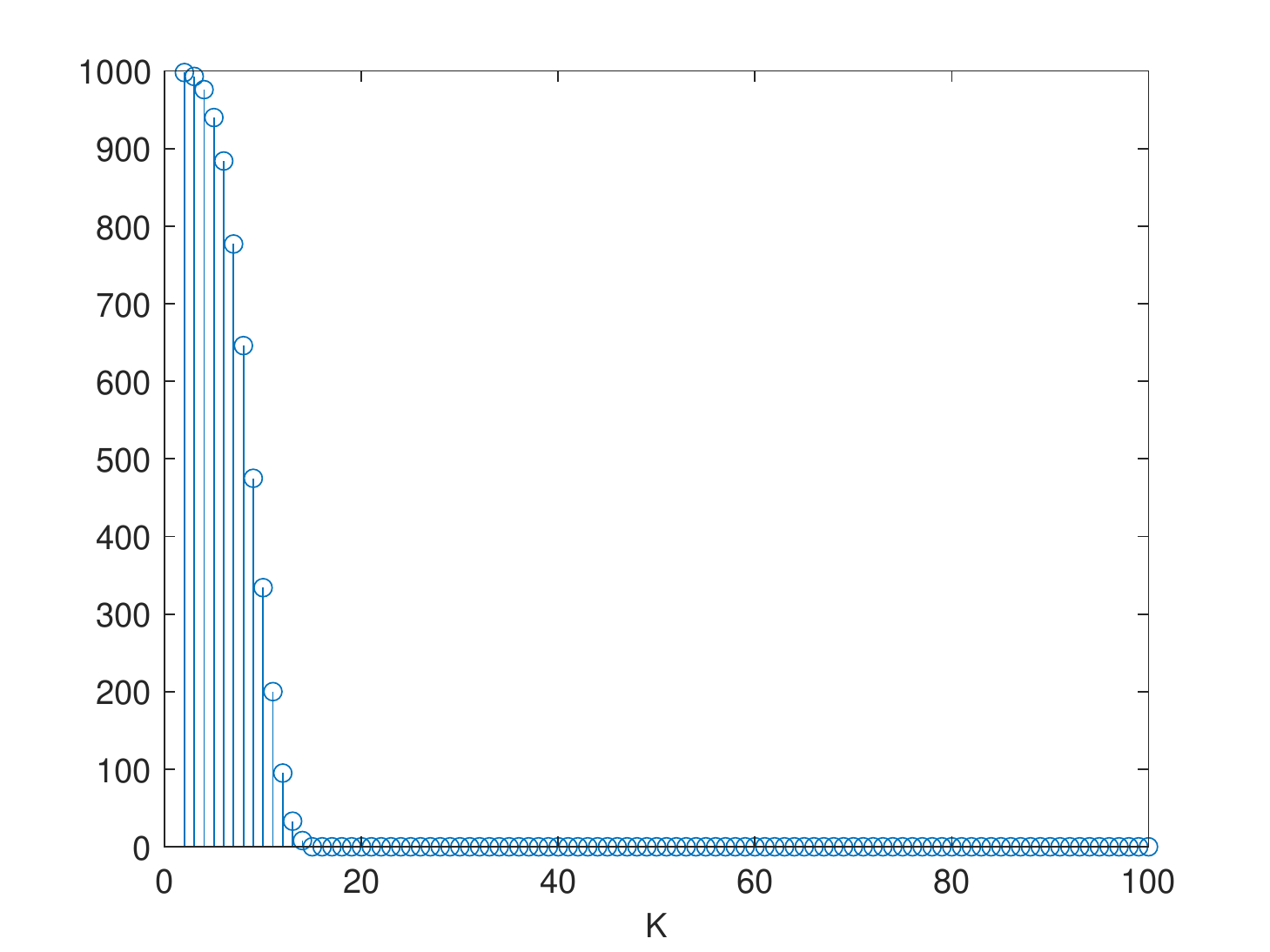}
		\caption{Number of collision-free events}
	\end{subfigure}
	\hfill
	\begin{subfigure}[h]{0.5\columnwidth}
		\centering
		\includegraphics[width=\columnwidth]{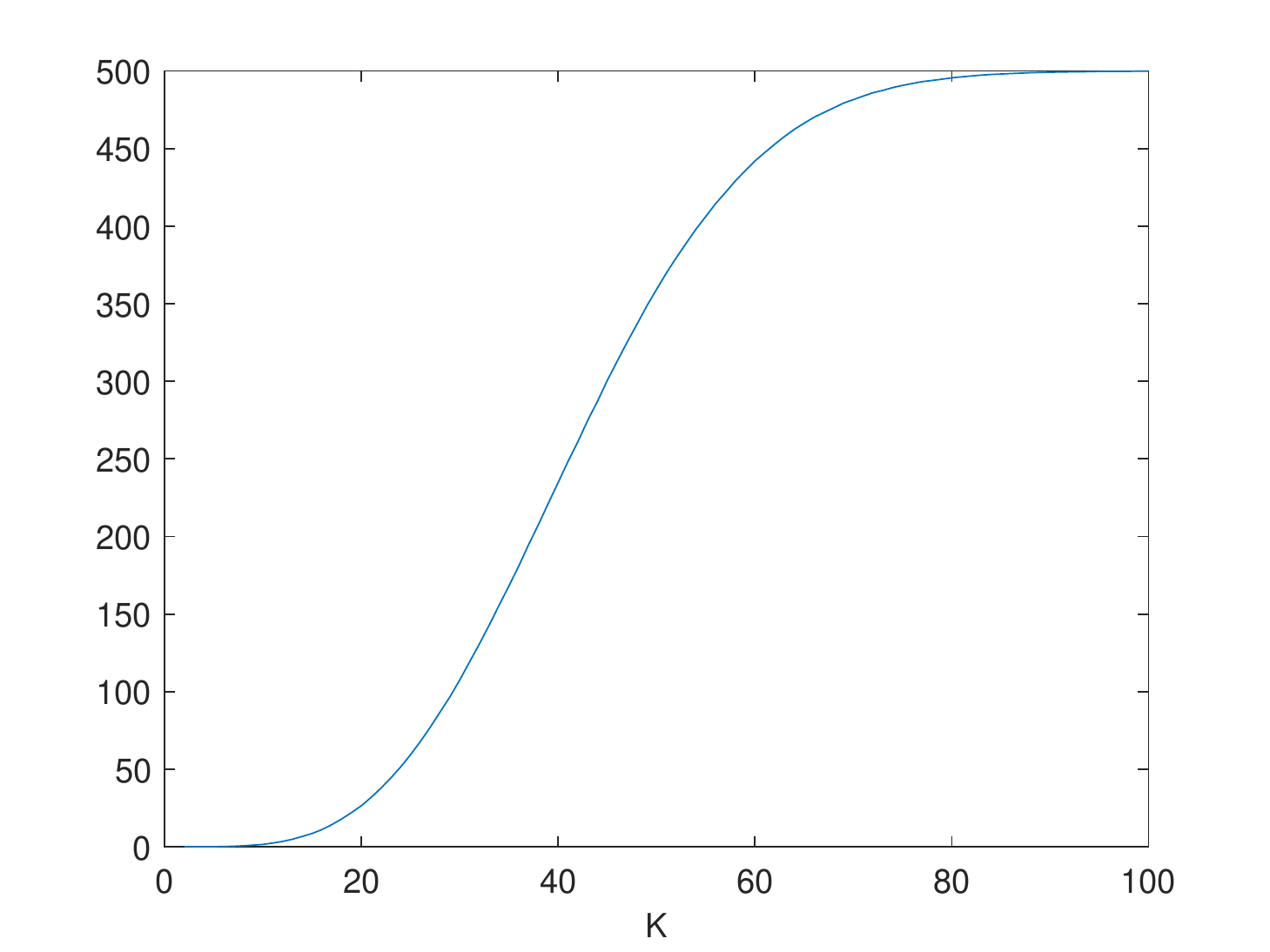}
		\caption{Average number of collisions}
	\end{subfigure}
	\caption{\label{fig:collision} For a fixed $K$, we generated 1000 K-sparse signals of length $N=1000$ and counted collisions (namely, entries of $|S-S|$ with multiplicity greater than one; see Section~\ref{sec.difference}). The left panel shows the number of collision-free events. The right panel presents the average number of collisions per trial. Clearly, unless $K/N$ is very small, collision-free events are rare.} 
\end{figure}

\begin{figure}[ht]
	\begin{subfigure}[h]{0.5\columnwidth}
		\centering
		\includegraphics[width=\columnwidth]{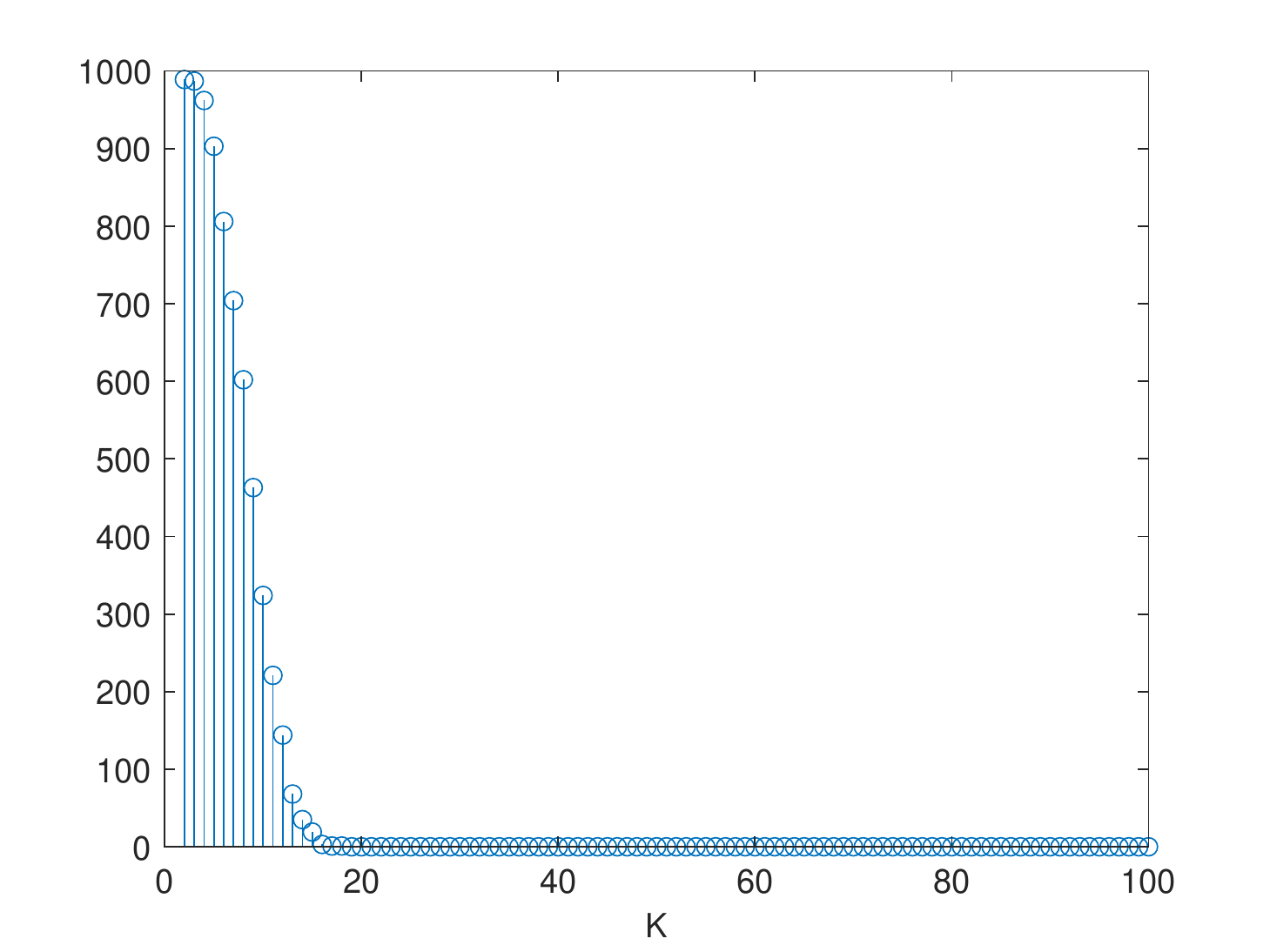}
		\caption{Number of collision-free events}
	\end{subfigure}
	\hfill
	\begin{subfigure}[h]{0.5\columnwidth}
		\centering
		\includegraphics[width=\columnwidth]{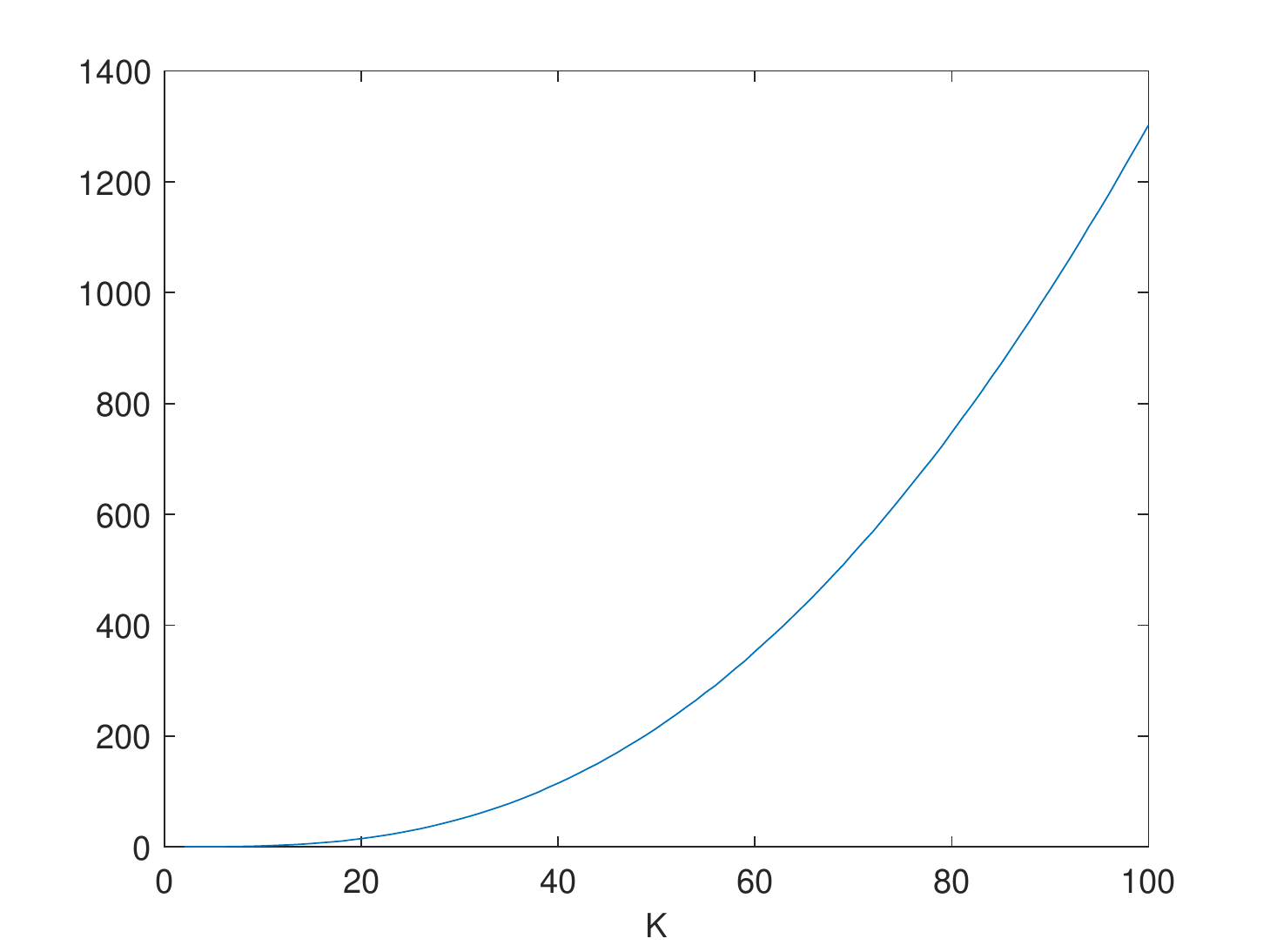}
		\caption{Average number of collisions}
	\end{subfigure}
	\caption{\label{fig:collision_fixed_ratio} For \rev{a} fixed $K$, we generated 1000 K-sparse signals while keeping a fixed ratio  $K/N=0.01$, and counted collisions. The left panel shows the number of collision-free events. The right panel presents the average number of collisions per trial. Clearly, even for fixed $K/N$ collision-free events are rare as $N$ grows.} 
\end{figure}

The crystallographic phase retrieval problem for binary signals depends solely on  difference multi-sets: two binary signals with
sparsity~$K$ have the same periodic auto-correlation if and only if they have the same
difference multi-sets.
Thus, the failure to distinguish non-equivalent binary signals from their
auto-correlation is equivalent to the existence of two non-equivalent
$K$-element subsets of $[0,N-1]$ with the same  difference multi-set.
For example, the subsets of
$[0,7]$, $\{0,1,3, 4\}$ and $\{0,1,2,5\}$ both have cyclic difference
multi-sets $\{0^4, 1^2, 2^1, 3^2, 4^1\}$ but are not equivalent.
Yet, these cases seem to be rare, suggesting that uniqueness for the binary case is ubiquitous. 

\subsection{Impossibility results}
We continue our investigation with some impossibility results.
We start with a simple parameter counting argument.
\begin{prop}
	A necessary condition for solving the
	crystallographic phase retrieval problem~\eqref{eq:problem} for generic signals is that $K = |S| \leq N/2 +1$. 
\end{prop}
\begin{proof}
	Since $|S-S|$ is the cardinality of the auto-correlation's support, then
	a parameter count implies that signal reconstruction is impossible
	if $|S-S| < |S|$. 
	Since $|S -S| \leq   N/2+1$, a necessary condition for solving the
	sparse phase retrieval problem is that $|S| \leq  N/2 +1$.
\end{proof}

We say that a subset $S \subset [0,N-1]$ is an {\em arithmetic progression
	with difference $d$} if there
exists $d \in [0, N-1]$ such that $S = \{c_0 + \ell d\,|\, \ell = 0,
\ldots |S|-1\} \MOD N$ \footnote{The
	term {\em periodic} has been used in~\cite{jaganathan2017sparse},
	but {\em arithmetic progression} is more consistent with arithmetic
	combinatorics literature~\cite{kemperman1960sumset}, where the term
	periodic is used only when $d$ divides $N$.}. 
Note that because the indices are taken  modulo $N$, we can always
assume that $d\leq N/2+1$.
For example, if $N = 9$ then the subsets
$\{0,2,4,6\}$ and $\{0,3,5,7\}$ are both arithmetic progressions
with $d=2$, where  $c_0 =0$ and $c_0 = 3$, respectively. 
\rev{The} property of being
an arithmetic progression is preserved by the action of the dihedral
group\rev{: if} $S$ is an arithmetic progression with difference $d$ and
$s \in D_{2N}$ is a reflection, then $s\cdot S$ is also arithmetic
progression with difference $d$.

If $S$ is an arithmetic progression, then $S-S = \{0,\overline{d}, \ldots ,
\overline{(K-1)d}\}$, where $\overline{m}$ denotes the equivalence
class of $n$ in $[0, N/2 ]$ under the
equivalence $n \sim -n$. If all of $\{0, \overline{d}, \ldots , \overline{(K-1)d}\}$
are distinct then $|S-S| = K$, but it is possible for $|S-S| < K$.
For example, if $S = \{0,2,4,6\} \subset [0,7]$ then
$S-S = \{0,2,4\}$ because $6 =(- 2) \bmod 8$.

\begin{prop} \label{prop.progression}
	Let $S$ be an arithmetic progression, and let $L_S$ be the vector
	space of signals whose support is $S$.
	Then, for  a generic vector $x\in L_S$ there is no unique solution to the crystallographic phase retrieval problem.
\end{prop}
\begin{proof}
	Applying a shift we can assume $0 \in S$
	so $S = \{0, d, 2d,\ldots d(K-1)\}$
	and $S-S$ is the set
	$\{0, \overline{d}, \overline{2d}, \ldots \overline{(K-1)d}\}$. 
        To simplify notation we assume that all of
        the $\overline{d \ell}$ are distinct so that $|S-S| = K$.
	In this case, if $x \in L_S$ 
	then the non-zero entries of the periodic auto-correlation of 
	$$x=(x[0], 0, \ldots , x[d], 0, \ldots , x[2d], 0 \ldots , 0,
	x[d(K-1)], 0 \ldots , 0),$$ are the same as the entries of the  aperiodic
	auto-correlation of the vector $$x'=(x[0], x[d], \ldots ,
	x[d(K-1)]) \in \K^K.$$ Precisely, we have
	\begin{equation} \label{eq.arith}
	a_x[\overline{\ell d}] = x[0]\overline{x[\ell d]} + x[d]
	\overline{x[(\ell+1)d]} + x[K-1-\ell d] \overline{x[(K-1)d]},
	\end{equation}
        where the $\overline{x[m d]}$ on the right-hand side indicates
        the complex conjugate, and the notation
        $a_x[\overline{\ell d}]$ indicates the entry indexd by
        the integer $\overline{\ell d} \in [0,  N/2 ]$.
	The right-hand side of~\eqref{eq.arith} is exactly the $\ell$-th
	entry in the aperiodic auto-correlation of the vector
	$x' = (x[0], x[d], \ldots , x[(K-1)d]) \in \K^K$, which does not determines a \rev{generic} signal uniquely~\cite{Bendory2017}.
\end{proof}
\rev{
\begin{remark}
  In our model the basic signal $x$ is periodically repeated to represent the crystal structure. If the support $S$ of $x$ is an arithmetic progression then the basic signal is itself periodic. In this case, Proposition~\ref{prop.progression}
  says that we cannot solve the phase retrieval problem for $x$. The reason
is that if we replace  $x$ by the signal $x_p$ representing one period of $x$ in $S$, then $x$ and $x_p$ have the same periodic repetition. 
However, the support of $x_p$ may no longer be sparse as occurs in Example \ref{ex.arith} below.
For this reason we cannot expect to be able recover $x_p$, or equivalently $x$,
  from its Fourier magnitude without further information. In practice we do not expect this situation to occur.
\end{remark}	
}

\begin{example} \label{ex.arith}
  Suppose that $S = \{0,2,4,6\} \subset [0,8]$,
  so $S$ is an arithmetic progression with $d =2$.
  The difference set is $\{\overline{0}, \overline{2}, \overline{4},\overline{6}\} = \{0,2,4,3\}$ since $6 = -3 \in \Z_9$. 
Then, the non-zero entries
	of the periodic auto-correlation $a_x$ are
	\begin{align*}
	a_x[0] &= |x[0]|^2+ |x[2]|^2 + |x[4]|^2 + |x[6]|^2 \\
	a_x[2]&= x[0] \overline{x[2]} +
	x[2] \overline{x[4]} + x[4]\overline{x[6]} \\
	a_x[3] &= x[6]\overline{x[0]} \\
	a_x[4] &= x[0]\overline{x[4]} + x[2]\overline{x[6]}.
	\end{align*}
	If we let $x' = (x[0], x[2],x[4],x[6]) \in \K^4$ and
	denote by $b_{x'}$ the aperiodic auto-correlation~\eqref{eq:non-periodic_ac}, 
	then $b_{x'}[0] = a_x[0]$, $b_{x'}[1] = a_x[2]$, $b_{x'}[2]=a_{x}[4]$,
        $b_{x'}[3] = a_x[3= \overline{6}]$.
\end{example}

\subsection{The main conjecture}
\subsubsection{Terminology from algebraic geometry}
We recall some terminology from algebraic geometry; for more detail see \cite{conca2015algebraic,edidin2017projections} and the references therein.

Let $\K$  denote either $\R$ or $\C$. Given polynomials
$f_1, \ldots , f_r \in \K[x_0, \ldots , x_{N-1}]$, let us define the set
$$Z(f_1, \ldots , f_r) = \left\{(a_0, \ldots a_{N-1})\,|\, f_i(a_0, \ldots , a_{N-1}) = 0, \quad {\text{for all}} \quad 
i =1, \ldots r\right\}.$$ A set of the form $Z(f_1, \ldots , f_r)$ is
called an {\em algebraic set}. The {\em Zariski topology} on $\K^N$ is the topology formed \rev{by defining open sets} to be the complements of algebraic sets.
Note that a Zariski closed
set is also closed in the Euclidean topology. The complement of an algebraic set
is called a {\em Zariski open} set. Every proper algebraic set in $\K^N$ has dimension
strictly smaller than $N$ and every non-empty Zariski open set $U$ is dense in both
the Zariski topology and the Euclidean topology. If $U$ is a non-empty
Zariski open set then $\K^N \setminus U$ has dimension strictly less than $N$
and therefore has Lesbegue measure 0.

\subsubsection{Statement of the main conjecture}
Recall that we denote by $L_S$ the subspace
of~$\K^N$ consisting of vectors whose support is contained in $S\subset [0,N-1]$. 
The following formulates Conjecture~\ref{conj.main.descriptive}|the main conjecture of this work|in more technical terms. 
\rev{Specifically, it states that if  the condition $|S-S| > |S|$ is met, then the $D$-orbit of a generic signal with support contained in $S$ is determined from its Fourier magnitude.}

\begin{conj} \label{conj.main.technical}
	Let $S$ be a subset of $[0, N-1]$ such that $|S-S| > |S|$, and let $x \in L_S$ be a generic signal. Then:
	\begin{itemize}
		\item if $a_x = a_{x'}$, then
		$x'$ is obtained from $x$ by an action of the group $D$ of intrinsic symmetries described in Proposition~\ref{prop:trivial_symmetries}; or equivalently,
		\item the Fourier magnitude mapping $x\mapsto |\hat{x}|$ is injective, modulo intrinsic symmetries.
	\end{itemize} 
\end{conj}
By generic signals, we mean that there is a non-empty Zariski open set $U_S \subset L_S = \K^{|S|}$
such that Fourier magnitude mapping $x \mapsto |\hat{x}|$ is injective modulo
\rev{intrinsic} symmetries at all points $x \in U_S$.

Although we cannot prove Conjecture~\ref{conj.main.technical}, we provide a computational
method to check the results for any given $K$ and $N$. There
are two aspects to verifying the conjecture: \rev{recovering the support of the signal, and signal
recovery for a given support. 
Each of these aspects requires a different computational verification test. Consequently, we treat them separately and formulate independent conjectures.} 
 We now elaborate about both. 

\subsubsection{Support recovery} \label{subsec.support_recovery}
The support of the periodic auto-correlation is
the set $S-S$.
If $S'$ is another $K$-element subset such that $S'-S' \neq S-S$, then the auto-correlation of a generic vector $x'$ in $L_{S'}$ has a different support than the auto-correlation of a generic vector $x$ in $L_S$. Thus, in order to investigate  recovery of generic $K$-sparse signals in $L_S$, we \rev{only need to} consider $K$-sparse subsets $S'$ with \rev{the same difference sets} $S' - S' = S-S$ as subsets of $[0, N/2]$. 

We denote by $a(L_S)$
the image of the subspace $L_S$ under the auto-correlation
map $\K^N \to \K^{N/2+1}$ for either $\K=\R$ or $\K=\C$. 
\rev{The following conjecture states that if $|S-S|>|S|$, then for generic signals the support set $S$ is determined, up to dihedral equivalence, from the Fourier magnitude.} 
\rev{
\begin{conj} \label{conj.support_recovery}
  Suppose that $S$ and $S'$ are two non-equivalent 
  $K$-element subsets of $[0,N-1]$ (i.e.,
  $S'$ is not in the orbit of $S$ under the action of the dihedral group)
	with $|S-S| =|S' -S'| > K$. Then, for generic $x \in L_S$, $a(x)$ is not in $a(L_{S'})$. Namely,  the support of $x$ is determined, up to dihedral equivalence, by the periodic auto-correlation of $x$.  
\end{conj}
}


\paragraph{Verifying Conjecture~\ref{conj.support_recovery} computationally} 
For a specific pair $S,S'$, there is a method to verify Conjecture \ref{conj.support_recovery} as follows.
Consider the incidence subvariety $I_{S,S'}$ of $L_S \times L_{S'}$
consisting of pairs $$\{(x,x')\,| \, a(x) = a(x'), x \in L_S, x'\in L_{S'}\}.$$
The projection $\pi_S \colon I_{S,S'} \to L_S$ 
is the set of $x \in L_S$ such that there exists $x' \in L_{S'}$
with $a(x) = a(x')$.  
To prove the conjecture, it suffices to prove
that $\pi_S(I_{S,S'})$ is not dense; this in turn implies that for a generic signal  if $a(x)=a(x')$ then $S=S'$.
For this statement, it is sufficient
to show that $\dim I_{S,S'} < |S|$. The reason
this is sufficient is that if $\dim I_{S,S'} < |S| = \dim L_S$ 
then $\dim \pi_S(I_{S,S'}) < \dim L_S$ as well, which means that the complement
$L_S \setminus \pi_S(I_{S,S'})$ is dense in the Zariski topology. Now,
if $x \in L_S \setminus \pi_S(I_{S,S'})$ then by definition there is no
$x' \in L_{S'}$ such that $a_x = a_{x'}$.
Since a finite intersection of Zariski dense subsets is Zariski
dense, we see that if $x$ is in the Zarski dense set $L_S \setminus
\bigcup_{\{S' | S'-S' = S-S\}}\pi_S(I_{S,S'})$ (namely, we consider all possible, finitely many, relevant cyclic difference sets) then there is no
non-equivalent subset $S'$ and vector $x' \in L_{S'}$ such that $a_x =
a_{x'}$. In other words, for generic $x$ in $L_S$ the equivalence class
of the subset $S$ is determined by the auto-correlation $a_x$.

As mentioned above, it suffices to check whether $\dim I_{S,S'} < |S|$.
When $\K = \R$ (the main interest of this paper) \rev{and $K$ is small}, the dimension of the variety $I_{S,S'}$ can be
computed \rev{relatively quickly}  using a computer algebra system to
compute the Hilbert polynomial of the ideal\footnote{Recall that the ideal generated by a set of polynomials is all polynomial combinations of its generators $f_1,\ldots,f_n,$: $I = \left\{\sum_{i=1}^n c_if_i \text{ for } c_i\in\K[x_1, \ldots , x_n]\right\}$.} defining $I_{S,S'}$ in the
polynomial ring $\K[\{x_i\}_{i \in S}, \{x'_j\}_{j \in S'}]$. \rev{(See Section \ref{sec.complexity}
for a discussion on the computational complexity of computing Hilbert polynomials.)}
The degree of the Hilbert polynomial is the dimension of the variety, so
a sufficient condition for the conjecture to hold is if
the degree of the Hilbert polynomial is less
than $|S|$. More technical details are provided in Appendix~\ref{sec:hilbert}.
\begin{example} \label{example.support_recovery.appendix}
	We give an explicit example to illustrate the methods used to generate the
	data presented in Example~\ref{example.support_recovery} below.
	Let $S = \{0,1,2,4\}$ and $S' = \{0,1,2,5\}$ be subsets
	of $[0,7]$. If $x = (x_0, x_1, x_2, 0, x_4, 0, 0, 0) \in L_S$
	and $y = (y_0, y_1, y_2, 0, 0, y_5, 0, 0) \in L_{S'}$ then
	$$a_x = (x_0^2 + x_1^2  + x_2^2 + x_4^2, x_0x_1 + x_1x_2, x_0x_2 + x_2 x_4, x_1x_4, x_4x_0),$$ and
	$$a_y  = (y_0^2+ y_1^2 + y_2^2 + y_5^2, y_0y_1+ y_1y_2 , y_0y_2,y_2y_5 + y_5y_0, y_1y_5).$$
	Hence $a_x = a_y$ if and only if the following five equations are
	satisfied:
	\begin{equation} \label{eq.incidence_eqs}
          \begin{array}{lcc}
	x_0^2 + x_1^2  + x_2^2 + x_4^2- y_0^2- y_1^2 - y_2^2 - y_5^2 & = & 0\\
	x_0x_1 + x_1x_2 -  y_0y_1- y_1y_2 & = & 0\\
	x_0x_2 + x_2 x_4 - y_0y_2 & = &0\\
	x_1 x_4 - y_2 y_5 - y_5 y_0 & = & 0\\
	x_0x_4 - y_1 y_5 & = & 0
	  \end{array}
          \end{equation}
	Thus, the incidence $I_{S,S'} = \{(x,y)\,|\, a_x = a_y \} \subset L_S \times L_{S'}$ is the algebraic subset of $\K^4 \times \K^4$ defined by the set of equations
	\eqref{eq.incidence_eqs}. Therefore, the generators of the ideal of	$I_{S,S'}$ are the five polynomials in the left-hand side of~\eqref{eq.incidence_eqs} included 
        in $\R[x_0,x_1,x_2,x_4, y_0,y_1,y_2, y_5]$.
	Using Macaulay2~\cite{M2}, we calculated the Hilbert polynomial of this ideal to
	be $32P_2-80P_1 + 80P_0$, which means that $I_{S,S'}$
	is a 3-dimensional algebraic subset of $\K^4 \times \K^4$
	and therefore its image under $\pi_S$ to $\K^4$ has dimension
	at most 3.
\end{example}

\begin{example} \label{example.support_recovery}
	We consider the case where $K =4$ and $N=8,9$.
	As presetned in Table~\ref{tab:1}, when $K =4, N=8$, there are 8 equivalence classes of $4$-sparse subsets of which only $4$ satisfy $|S-S| > 4$. 	For the 6 pairs of subsets $S,S'$ with $S-S =S'-S' = \{0,1,2,3,4\}$, we have
	verified using Macaulay2~\cite{M2} that $\dim I_{S,S'} = 3$; 
	this is the desired result since $|S-S| =5$ means that we expect to impose 5 constraints on the
	8-dimensional space $L_S \times L_{S'}$.
	Hence, the support of \rev{a} generic vector $x \in L_S$ can be recovered
	from its periodic auto-correlation.
	However, this is not the case if
	$|S-S| = 4$. For example, if $S = \{0,1,2,3\}$ and $S' = \{0,1,3,6\}$
	then $\dim I_{S,S'} = 4$. 
	
\begin{table}[h]
	\begin{center}
		\caption{Verification of Conjecture~\ref{conj.support_recovery} for $N=8, K=4$}
		\label{tab:1}
	\begin{tabular}{c|c|c}
		$S$ & $S-S$ & $|S-S|$\\
		\hline
		\{0,1,2,3\} & \{0,1,2,3\}  & 4 \\
		\{0,1,2,4\} & \{0,1,2,3,4\} & 5 \\
		\{0,1,2,5\} & \{0,1,2,3,4\} & 5\\
		\{0,1,3,4\} & \{0,1,2,3,4\} & 5\\
		\{0,1,3,5\} & \{0,1,2,3,4\} & 5\\
		\{0,1,3,6\} & \{0,1,2,3\} & 4\\
		\{0,1,4,5\} &  \{0,1,3,4\} & 4\\
		\{0,2,4,6\} & \{0,2,4\}& 3\\
		\hline
	\end{tabular}
\end{center}
\end{table}

	Recall that when $S-S$ and $S'-S'$ differ as multi-sets, we can recover the support of a binary signal from its periodic
	auto-correlation (and thus, obviously, also the binary signal itself). Interestingly, the non-equivalent subsets $\{0,1,3,4\}$
	and $\{0,1,2,5\}$ have the same difference multi-sets so we cannot
	distinguish their supports from binary signals but we can from generic signals.

	When $N =9$ and $K=4$, there are 10 equivalence classes
	of subsets, which are presented in Table~\ref{tab:2}. In this case, all incidences $I_{S,S'}$ are 3-dimensional, so
	the support of \rev{a} generic vector can be determined from its auto-correlation even if $|S-S| = 4$ because 
		for each distinct pair $S,S'$ $|(S-S) \cup (S'-S')| \geq 5$
                which means that we obtain at least five constraints
                  on the entries of a pair $(x,y) \in I_{S,S'}$.
                  For example if $S = \{0,1,2,3,\}$ and
                  $S' = \{0,2,4,6\}$ and $x = (x_0, x_1, x_2, x_3,0,0,0,0,0)\in L_{S}$ and $y = (y_0, 0, y_2, 0, y_4, 0, y_6, 0,0) \in L_{S'}$
                  then $a_x = a_y$ if and only if the following
                  5 equations are satisfied
	\begin{equation} 
          \begin{array}{lcc}
	x_0^2 + x_1^2  + x_2^2 + x_3^2- y_0^2- y_2^2 - y_4^2 - y_6^2 & = & 0\\
	x_0x_1 + x_1x_2 + x_2 x_3  & = & 0\\
	x_0x_2 + x_1 x_3 - y_0y_2 -y_2y_4 - y_4 y_6 & = &0\\
	x_0 x_3 - y_6 y_0 & = & 0\\
        y_0 y_4 + y_2y_6  & = & 0
	  \end{array}
          \end{equation}
		Namely, in this specific example we do not demand $|S-S|>K.$
		
	\begin{table}[h]
	\begin{center}
		\caption{Verification of Conjecture~\ref{conj.support_recovery} for $N=9, K=4$}
\label{tab:2}
		\begin{tabular}{c|c|c}
			$S$ & $S-S$ & $|S-S|$\\ 
			\hline
			\{0,1,2,3\} & \{0,1,2,3\} & 4\\
			\{0,1,2,4\} & \{0,1,2,3,4\}& 5\\
			\{0,1,2,5\} & \{0,1,2,3,4\} & 5\\
			\{0,1,3,4\} &  \{0,1,2,3,4\}& 5\\
			\{0,1,3,5\} & \{0,1,2,3,4\}& 5\\
			\{0,1,3,6\} & \{0,1,2,3,4\}& 5\\
			\{0,1,3,7\} & \{0,1,2,3,4\}& 5\\
			\{0,1,4,5\} & \{0,1,3,4\}& 4\\
			\{0,1,4,6\} & \{0,1,2,3,4\}& 5\\
			\{0,2,4,6\} & \{0,2,3,4\}& 4\\
			\hline
		\end{tabular}
	\end{center}
	\end{table}

\end{example}

\subsubsection{Generic signal recovery given knowledge of the support}\label{subsec.vector_recovery}
The difficulty of reconstructing a signal given its support depends on the structure of $S-S$.
We first consider two simple cases, and then move forward to the general case. 

The easiest case is the collision-free case. This means that no non-zero entry in the cyclic difference set $S-S$ appears with multiplicity greater than one. In this case there is a relatively easy eigenvalue argument to determine the
entries of $x$ from its auto-correlation~\cite{ranieri2013phase}.
Note, however, that collision-free subsets
appear to be quite rare unless $K$ is very small compared to $N$, as demonstrated in Figure~\ref{fig:collision}. 
Even if we keep the ratio $K/N$ fixed, then \rev{collision-free events} are rare as $N$ grows. 
Figure~\ref{fig:collision_fixed_ratio} exemplifies it for $K/N=0.01$, which is the expected sparsity level in proteins.

The other easy case is when the difference set $S$ can be concentrated
in the interval $[0, N/2]$ after applying reflection and
translation (i.e., an element of the dihedral group). In this case, the periodic auto-correlation of $x \in L_S$ is
the same as the non-periodic auto-correlation of $x$ viewed as a
vector in $\R^{N/2+1}$.  For example, if $N =8$ then the support set
$\{0, 5,7\}$ can be moved to $\{0,1,3\}$ by reflection and this set
is concentrated.
For concentrated sets uniqueness of recovery depends on whether the
support forms an arithmetic progression as discussed in~\cite{jaganathan2012recovery}.

Next, we consider the general case.
Given a subset $S \subset [0,N-1]$, we let $D_S$ be the subgroup
of the group of intrinsic symmetries $D$ that leaves $L_S$ invariant. 
We refer to this group as the \emph{group of  intrinsic symmetries of $S$}.  
For a typical $S$, $D_S = \pm 1$ if $\K = \R$ and $D_S = S^1$ if $\K = \C$ (recall that these subgroups do not affect the support).
However, there are subsets $S$ for which $D_S$ is a bigger group. 
For example, the subset $\{0,1,3,4\}$ of $[0,7]$ is preserved by the subgroup of $D_{16}$ consists of two elements: the identity
and reflection followed by four shifts.
Thus, $D_S = \pm 1 \times \Z_2$ or $D_S = S^1 \times
\Z_2$, depending on whether $\K$ is real or complex.

\rev{The following conjecture states that if the support of the signal is known, the Fourier magnitude determines a generic signal uniquely, up to an element of $D_S$.}
\begin{conj} \label{conj.vector_recovery}
	Suppose that $|S - S| > |S|$.  If $x \in L_S$ is a generic
	vector and $x' \in L_S$ is another vector (in the same subspace) such that $a(x) = a(x')$, then
	$x' = g \cdot x$ for some $g \in D_S$.
\end{conj}

Conjecture~\ref{conj.support_recovery} states that if $|S -S| > |S|$
then we can recover the support
$S$ of a generic vector in $L_S$. Conjecture~\ref{conj.vector_recovery} argues that
once we know that $x \in L_S$, then we recover $x$ itself. These two conjectures combined imply Conjecture~\ref{conj.main.technical}.

\paragraph{Verifying Conjecture~\ref{conj.vector_recovery} computationally} 
To verify the conjecture we consider the incidence
\rev{\begin{equation} \label{eq.incidence}
I_S = \{(x,x') \,| \,a(x) = a(x')\} \subset L_S \times L_S.
\end{equation}}
By construction, $I_S$ contains the $K$-dimensional linear subspaces $L_g = \{(x,g\cdot x)\, | \, g \in D_S\}$. The goal is to show that $I_S \setminus \cup_{g \in D_S} L_g$ (namely, the incidence without the subspaces corresponding to the
intrinsic symmetries of $S$) has dimension
strictly less than $K$.
To verify the conjecture we can show that the $K$-dimensional
components of $I_S$ correspond to pairs $(x,x')$ where $x'$ is obtained from
$x$ by a trivial ambiguity and that all other components have strictly smaller dimension.

This can be done by computing the Hilbert polynomial
$P_I$
of the ideal $I$ which defines the algebraic subset $I_S \subset L_S \times L_S$. 
As discussed in Appendix~\ref{sec:hilbert}, the Hilbert polynomial
$P_I$ can expressed in the following form
$$P_I = a_\ell P_\ell + a_{\ell -1}P_{\ell -1} + \ldots a_0 P_0,$$
where $\ell = \dim I_S -1$ and $a_\ell$ is the degree of
 $I_S$ as an algebraic subset of $L_S \times L_S$.
Since a linear subspace has degree one,
to show that $I_S \setminus \cup_{g \in D_S} L_g$ has dimension smaller than $K$, it suffices to show that $\ell \leq |S| -1$
and
$a_\ell = |D_S|$ 
whenever $|S - S| > K$.
\begin{example}
We give an example to illustrate the methods used to generate the
	data presented in Example~\ref{example.vector_recovery} below.
        Let $S = \{0,1,2,5\} \subset [0,7]$ and let
        $L_S$ be subspace of $\R^N$ with support in $S$. The set $S$
        is preserved by the element \rev{of order two} $\sigma \in D_{16}$
        which is the reflection composed with a shift by 2. Thus, the
        group $D_S$ that stabilizes $L_S$ consists
        of four elements which we denote by $(1,-1,\sigma, -\sigma)$. Specifically, if $(a,b,c,0,0,d,0,0) \in L_S$ then:
        \begin{align*}
1\cdot (a,b,c,0,0,d,0,0) \cdot &= (a,b,c,0,0,d,0,0) \\
-1\cdot(a,b,c,0,0,d,0,0)&=(-a,-b,-c,0,0,-d,0,0) \\
        \sigma\cdot (a,b,c,0,0,d,0,0)&= (c,b,a,0,0,d,0,0) \\
        -\sigma\cdot (a,b,c,0,0,d,0,0)& =(-c,-b,-a,0,0,-d,0,0).
        \end{align*}
                 If $x=(x_0, x_1, x_2, 0, 0, x_5, 0, 0)$
        and $y = (y_0, y_1, y_2, 0, 0, y_5, 0, 0)$, then
        $a_x = a_y$ if and only if the following  equations are satisfied:
        \begin{equation} \label{eq.incidence_eqs2}
          \begin{array}{lcc}
	x_0^2 + x_1^2  + x_2^2 + x_5^2- y_0^2- y_1^2 - y_2^2 - y_5^2 & = & 0\\
	x_0x_1 + x_1x_2 -  y_0y_1- y_1y_2 & = & 0\\
	x_0x_2  - y_0y_2 & = &0\\
	x_2x_5 + x_5 x_0 - y_2 y_5 - y_5 y_0 & = & 0\\
	x_1x_5 - y_1 y_5 & = & 0
	  \end{array}
        \end{equation}
        Let $I$ be the ideal in $\K[x_0,x_1, x_2, x_5, y_0, y_1, y_2, y_5]$
        generated by the five polynomials in the left-hand side of~\eqref{eq.incidence_eqs2}.
        The equations~\eqref{eq.incidence_eqs2} are clearly satisfied 
        if $x = y$ of $x = -y$. Thus, the 4-dimensional linear subspaces
        $L_1 = \{(x,x)\,|\,x \in L_S\}$ and $L_{-1} = \{(x,-x)\,|\, x \in L_S\}$
        are in $Z(I)$, where $Z(I)$ denotes the algebraic subset
        of $L_S \times L_S$ defined by the ideal $I$.
        In addition for any 4 real numbers $(a,b,c,d)$ the vectors
        $x = (a,b,c, 0, 0, d,0,0)$ and $y = \pm (c,b,a,0,0,d,0,0)$
        are solutions the equations \eqref{eq.incidence_eqs2}.
        Hence, there are two additional 4-dimensional linear subspaces
        $L_\sigma = \{ (x, \sigma x)\,|\, x\in L_S\}$ and $L_{-\sigma} = \{ (x,-\sigma x)\,|\,x \in L_S\}$ in $Z(I)$.

        Using Macaulay2 we calculated the Hilbert polynomial of
        the ideal $I$ to be $P_I = 4P_3 + 10P_2 -30P_1 + 20P_0$.
        Since $I_S  = Z(I)$ contains four linear subspaces
        $L_1, L_{-1}, L_{\sigma}, L_{-\sigma}$, then Proposition~\ref{prop.linear_strip} implies that 
        \begin{equation} \label{eq:example_known_support}
I_S \setminus (L_0 \cup L_{-1} \cup L_{\sigma} \cup L_{-\sigma})=
\{(x,x') \,|\, a_x = a_{x'} \text{ and } x' \neq g\cdot x \text{ for some } g \in D_S\},
        \end{equation}
has dimension at most 3. Hence, for generic $x \in L_S$
if $a_x = a_{x'}$ then $x' = g \cdot x$ for some $g \in D_S$. \end{example}

\begin{example} \label{example.vector_recovery}
	For each of the equivalence classes of 4-element subsets $S$ with $|S -S| > 4$ and
	$N=8, 9,$ we used Macualay2 \cite{M2} to compute the Hilbert
	polynomial of $I_S$ and verify that generic
	$x \in L_S$ can determined from its periodic  auto-correlation.
	Tables~\ref{table:3} and~\ref{table:4} present
	the degree and dimension of $I_S$.
	In each case, the degree of $I_S$ equals $|D_S|$ and dimension equals~$K$. 
	\begin{table}
\caption{Verification of Conjecture~\ref{conj.vector_recovery} for $N=8, K=4$}
	\label{table:3}
	\begin{center}
		\begin{tabular}{c|c|c|c|c}
			subset & $|D_S|$ & degree & dimension & phase retrieval\\
			\hline
			\{0,1,2,4\} & 2 & 2&4 & Yes\\
			\{0,1,2,5\} & 4  & 4&4 & Yes\\
			\{0,1,3,4\} & 4 & 4&4 & Yes\\
			\{0,1,3,5\} & 2 & 2&4 & Yes\\
			\hline
		\end{tabular}
	\end{center}
\end{table}
\begin{table}
\caption{Verification of Conjecture~\ref{conj.vector_recovery} for $N=9, K=4$}
\label{table:4}
	\begin{center}
		\begin{tabular}{c|c|c|c|c}
			subset & $|D_S|$ & degree& dimension & phase retrieval\\
			\hline
			\{0,1,2,4\} & 2 & 2&4 & Yes\\
			\{0,1,2,5\} & 2 & 2&4& Yes \\
			\{0,1,3,4\} & 4 & 4&4& Yes \\
			\{0,1,3,5\} & 2 & 2&4 & Yes\\
			\{0,1,3,6\} & 6 &  6&4& Yes \\
			\{0,1,3,7\} & 4 & 4&4 & Yes\\
			\{0,1,4,6\} & 4  & 4&4 & Yes\\
			\hline
		\end{tabular}
	\end{center}	
\end{table}

\end{example}
The hypothesis that $|S-S| > K$ in Conjecture~\ref{conj.vector_recovery} is necessary.
To demonstrate it, the following example gives two different signals with the same support and the same autocorrelation when $|S|=|S-S|=4$.
\begin{example}
	Consider the subset $S = \{0,1,4,5\}$ of $[0,7]$ so that
	$|S - S| = |S|=4$. Then, the vectors $$[x[0], x[1], 0,0,x[4],x[5],0,0],$$ and $$1/2[(x[0]+x[1] -x[4] + x[5]),
	(x[0] + x[1] + x[4] -x[5]), 0,0,(-x[0] + x[1] + x[4] + x[5]), (x[0] -x[1]+x[4] + x[5]),0,0)],$$ have
	the same auto-correlation but are not related by an intrinsic symmetry.
\end{example}
\rev{
\subsubsection{The computational complexity of verifying Conjecture~\ref{conj.support_recovery} and Conjecture~\ref{conj.vector_recovery}}
\label{sec.complexity}
  There is no expectation that the computational complexity of verifying conjectures~\ref{conj.support_recovery} and~~\ref{conj.vector_recovery} is polynomial in $N$ or in~$K$. There are two significant issues. The first is that verification of Conjecture \ref{conj.support_recovery}
  requires enumerating over all ${N \choose{K}}$ element subsets of $[1,N]$. If $K \sim N$ then Stirling's formula implies that this number asymptotic to at least $a^N/\sqrt{2\pi N}$ for some $a > e$.

In addition 
there are no good bounds on the computational complexity of computing the Hilbert polynomial of an ideal in
  a polynomial ring. The reason is that implemented algorithms for computing Hilbert polynomials first compute Gr\"obner bases. The
  computational complexity of computing Gr\"obner bases is not known, but
  for the ideals we consider which are generated by degree two elements in
  $2K$ variables, the best theoretical bound on the complexity is doubly exponential, namely $O(a^{2^K})$
  \cite{dube1990structure}\footnote{The doubly exponential bound of \cite{dube1990structure} is a bound on the maximum degree of an element in a Gr\"obner basis.}. To illustrate this, we tabulate below the run times of the {\tt hilbertPolynomial} function in Macaulay2~\cite{M2} to compute the Hilbert polynomial
  of the ideal of the incidence $I_S$ defined by \eqref{eq.incidence}, where $S$ is a random $K$ element subset of $[0,99]$, 
  and $K =5,\ldots, 12$.
  For these reasons numerical verification is only feasible for small-scale problems and cannot be applied directly to X-ray crystallography.}
\begin{table}[h]
	\begin{center}
	  \caption{Run times to compute the Hilbert polynomial
            of the ideal of $I_S$ for $S$, a random subset of $[0,99]$.}

		\begin{tabular}{c|c}
			$K$ & {\text Run time in seconds}\\ 
			\hline
			5 & 0.394786  \\
			6 & 0.300032\\
			7 &  0.48212\\
			8 & 1.02593 \\
			9 &  2.79231 \\
			10 & 38.6528 \\
			11 & 67.4881\\
                        12 & 191.163\\
			\hline
		\end{tabular}
	\end{center}
\end{table}



\section{Group-theoretic considerations}
\label{sec:group_theory}

In Fourier domain, signals $x$ and $x'$ have the same Fourier magnitude 
if and only if $\hat x'[i] = e^{\iota \theta_i} \hat x[i]$
for some set of rotations~$(\theta_0,\ldots,\theta_{N-1})$.
It follows that if $\K = \C$ then the group $G = (S^1)^N$
preserves $|\hat x|$. The group $G$ acts on signals
in the time domain via the Fourier transform. In other words,
if $x \in \C^N$ and $g = (e^{\iota \theta_0}, \ldots, e^{\iota \theta_{N-1}})$
then $g \cdot x = F^{-1} g Fx$, where $F$ is the $N\times N$ DFT matrix.

We call $G$ the group of {\em non-trivial symmetries} for the phase retrieval problem. The action of $G$ is related to the action of $D$, the group of
intrinsic symmetries  (see Proposition~\ref{prop:trivial_symmetries}), as follows. The subgroup $S^1$ of $D$
corresponds to the diagonal subgroup of $G = (S^1)^N$ since
$\widehat{e^{\iota \theta} x} = e^{\iota \theta}  \hat{x}$. The circular shift
of $D$ forms the subgroup $\Z_N \subset G$,
generated by the element $(1, \omega, \omega^2, \ldots , \omega^{N-1})$, 
where $\omega = e^{2\pi \iota/N}$. If $\tilde{x}$ is the reflected signal,
then in Fourier domain  $\hat{\tilde{x}} = \overline{\hat{x}}$. Thus, the action
of the reflection in $D_{2N}$ does not correspond to the action
of an element of $(S^1)^N$. However, by letting $\hat x = (e^{\iota \theta_0},
\ldots , e^{\iota \theta_{N-1}}) \odot |\hat{x}|$ we see that 
$\hat{\tilde{x}} = \overline{\hat{x}} = g_x \hat{x}$
where $g_x = (e^{-2\iota \theta_0},
\ldots e^{-2\iota \theta_{N-1}})$. It follows
that $\tilde{x}$ is in the $G$-orbit of $x$. Hence
the orbit $G\cdot x$ contains the orbit $D\cdot x$ even though the non-abelian group $D$ is not a subgroup of the abelian~$G$.

A similar analysis holds in the real case (as in crystallographic phase retrieval) but the group $G$ of non-trivial symmetries is smaller. The reason
is that if $g$ is an arbitrary element of $(S^1)^N$ then 
$g \cdot \hat x$ is not the Fourier transform of a real vector, because $x$
is real if and only $\hat x$ is invariant under reflection and conjugation;
i.e., $\hat x[N-i] = \overline{\hat x[i]}$. In particular, $\hat x [0]$
is real and if $N$ is even then $\hat x[N/2]$ is real as well.
Thus, if $\K = \R$ we let $G$ be the subgroup of $(S^1)^N$ that 
 preserves the Fourier transforms of real
vectors: $G = \{(e^{\iota \theta_0}, \ldots,
e^{\iota \theta_{N-1}} )\,|\, \theta_{n} + \theta_{N-n} \equiv 0 \MOD 2\pi\}$.
If $N$ is odd, then $G$ is isomorphic 
to  $\pm 1\ \times (S^1)^{\lfloor N/2\rfloor}$
and if $N$ is even 
$G$ is isomorphic to $\pm 1 \times (S^1)^{N/2 -1}  \times \pm 1$.
Again, if $x \in \R^N$ then the orbit $D\cdot x$ is contained in
the orbit $G\cdot x$.

\subsection{Group-theoretic formulation of Conjecture~\ref{conj.main.technical}}
Given a $K$-dimensional subspace~$L$ (not necessarily
sparse), we
denote by $G \cdot L$  the orbit of $L$ under the group $G$.
By definition, $G\cdot L = \{g\cdot x \,| \,g \in G, x \in L \} \subset \K^N$
and consists of all vectors $x' \in \K^N$ with the property
that $a_{x'} = a_x$ for some fixed $x \in L$.
If $S'$ is equivalent to $S$, then $G \cdot L_S = G \cdot L_{S'}$ because
$L_{S'} = d \cdot  L_S$ for some $d \in D$, where $D$ is the group of intrinsic symmetries. 
We can now reformulate our conjectures in group-theoretic terms.

The group-theoretic version of Conjecture \ref{conj.main.technical} can
be stated as follows. 
\begin{conj} \label{conj.main_group}
	
	Suppose that  $x \in L_S$ is a K-sparse generic signal  such that $|S-S| >K$. Then, the orbit $G\cdot x$ contains
	a single $D$ orbit, which corresponds to the intrinsic symmetries of a  $K$-sparse signal.
\end{conj}
Similarly, Conjectures~\ref{conj.support_recovery} and~\ref{conj.vector_recovery} are restated
as follows.
\begin{conj}[Support recovery] \label{conj.support_recovery_group}
	Suppose that  $x \in L_S$ is a K-sparse generic signal   such that $|S-S| >K$.
	Then, $D \cdot L_S$ is the only $D$ orbit of a linear  subspace of dimension $K$
	contained in $G \cdot L_S$. 
\end{conj}

\begin{conj}[Generic signal recovery] \label{conj.genericrecovery-group}
	Suppose that  $x \in L_S$ is a K-sparse generic signal   such that $|S-S| >K$. Then, $G\cdot x \cap L_S = D_S\cdot x$, where $D_S$ \rev{is the} group of  intrinsic symmetries of~$S$.
\end{conj}

\subsection{Conjecture: Sparse signals are rare among signals with the same auto-correlation}

Given the group-theoretic formulation of the crystallographic phase retrieval problem, we pose
an additional conjecture, stating that the set of K-sparse signals among all signals with the same periodic auto-correlation is of measure zero.
More precisely,
if $x$ is any $K$-sparse signal with 
$K\leq N/2 +1$, then
for a generic element $g$ in the group $G$ of non-trivial symmetries, $g\cdot x$ is
not $K$-sparse.   
The conjecture implies that even if there exist additional solutions to the crystallographic phase retrieval problem, they are of measure zero. 
Importantly, this conjecture refers to all signals, not necessarily generic, without imposing any structure on the signal's support. 


\begin{conj}[Generic transversality] \label{conj.transversality}
	Let $L_S$ be a $K$-sparse subspace of $\K^N$ (either $\K=\R$ or $\K=\C$).
	For generic $g$ in the group $G$ of non-trivial symmetries
        the following holds:
        \begin{enumerate}
          \item If $\K = \R$ and $K < N/2$, then for all $K$-sparse subspaces $L_{S'}$ (including
            $S' = S$) the translated subspace 
$g \cdot L_S$ has 0-intersection with $L_{S'}$ (i.e., $g\cdot L_S \cap L_{S'} = \{0\}$).
            \item If $\K = \C$ and $K \leq N/2$ then for all $K$-sparse subspaces $L_{S'}$ (including $S' = S$) the translated subspace 
              $g \cdot L_S$ has 0-intersection with $L_{S'}$.
         \end{enumerate}     
        In particular, if $K  < N/2$ then a generic translate of a sparse subspace
        contains no sparse vectors.
\end{conj}

\paragraph{Verifying Conjecture~\ref{conj.transversality} computationally.}
Given a $K$-element subset $S \subset [0,N-1]$, 
we can verify Conjecture~\ref{conj.transversality} as follows.
If $S = \{i_1, \ldots , i_K\}$
we let  $e_{i_1} \ldots , e_{i_K}$  be the standard basis for $L_S$, namely, 
$e_{i_j}$ denotes the vector $(0,0, \ldots 0, 1,0\ldots 0)$, where
the $1$ is in the $i_j$-th place.
Then, $g\cdot e_{i_1}, \ldots , g\cdot e_{i_K}$ form a basis for $g\cdot L_S$.
Let $S' = \{j_1, \ldots , j_k\}$ be any other $K$-element subset.
Then, $e_{j_1}, \ldots , e_{j_k}$ form a basis for $L_{S'}$,
and the $2K$ vectors $\{g \cdot e_{i_1}, \ldots , g \cdot e_{i_K}, e_{j_1}, \ldots , e_{j_K}\}$ span the subspace $g\cdot L_S + L_{S'}$ of $\K^N$.
By the standard linear algebra formula
\begin{equation*}
\dim (g\cdot L_S + L_{S'}) = \dim g\cdot L_S + \dim L_{S'} - \dim(g\cdot L_S \cap L_{S'}),
\end{equation*}
and thus we have $g\cdot L_{S} \cap L_{S'} = \{0\}$ if and only if $\dim (g\cdot L_{S} + L_{S'})
= 2K$. This is equivalent to requiring that the $2K$ vectors
$\{g \cdot e_{i_1}, \ldots , g \cdot e_{i_K}, e_{j_1}, \ldots , e_{j_K}\}$
be linearly independent.
Therefore,  $g\cdot L_S \cap L_{S'} \neq \{0\}$ if and only
if the $2K \times N$ matrix
\begin{equation*}
A_{S,S'}(g) = \begin{bmatrix}
e_{j_1}, 
e_{j_2},
\dots,
e_{j_K},
g\cdot e_{i_1},
\dots,
g\cdot e_{i_K}
\end{bmatrix}^T,
\end{equation*}
spanned by the $2K$ vectors $e_{i_1}, \ldots , e_{i_K}, g\cdot e_{i_1},
\ldots , g\cdot e_{i_K}$ (where we treat the vectors as row vectors)
has rank strictly less than $2K$.
  \begin{prop}
    If for each $K$-elements subset $S'$
    there exists a single $g_{S'}$ in each connected
    component\footnote{When $\K = \C$ the group
      $G$ of non-trivial symmetries is $(S^1)^N$, which is connected.
      However, if $\K = \R$ and $N$ is odd then $G = \Z_2 \times (S^1)^{\lfloor N/2 \rfloor}$
      and if $N$ is even then $G = \Z_2 \times (S^1)^{N/2-1} \times \Z_2$.
      Thus, if $\K = \R$ the group of non-trivial symmetries has either 2 or 4    connected components.}
    of $G$ such that $A_{S,S'}{(g_{S'})}$ has maximal rank,
     then for generic $g \in G$
     and all $K$-element subsets $S'$, $g \cdot L_{S} \cap L_{S'} = \{0\}$.
  \end{prop}
  \begin{proof}
  The first  $K$ rows of the matrix $A(g_{S'})$ are fixed,  while the last
$K$ rows depend linearly on the coordinates of $g_{S'} \in G$. The matrix
  $A_{S,S'}(g)$ fails to have rank $2K$ if and only if all $2K \times 2K$
  minors 
vanish. Each minor is polynomial in the entries of $A_{S,S'}(g_{S'})$
and thus a polynomial in the coordinates of $g_{S'}$. Hence, 
the set $Z_{S,S'} = \{g \in G\,|\, \rank A_{S,S'}(g) \text{ is not maximal}\}$ of matrices which do not have maximal rank is 
an algebraic subset of the real algebraic group~$G$. The set $U =G \setminus
\bigcup_{S'} Z_{S,S'}$
is Zariski open and consists of the $g \in G$ such that $g\cdot L_S$
is transverse to all $L_{S'}$.
Thus, to verify  Conjecture~\ref{conj.transversality}
for a specific $L_S$ it suffices to prove
that the intersection of $U$ with each connected component
is non-empty, implying it is dense.
 In other words, it suffices to find for each subset $L_{S'}$ a single
$g_{S'}$ in each connected component of $G$ such that $g_{S'}\cdot L_S$ is transverse to $L_{S'}$.
  \end{proof}

\begin{example}
	Using the technique above we verified Conjecture~\ref{conj.transversality}
	for every $K$-sparse subspace of $\C^{2K}$ with $2 \leq K \leq 7$.
	 We chose for each $K$ a random element
	$g_K \in G$ and showed that for each pair of $K$-element
	subsets of $[0, 2K-1]$ the appropriate matrix had maximal rank.
        Likewise, we verified the conjecture for every $K$-sparse
        subspace of $\R^{2K+1}$ for $2 \leq K \leq 7$. In this
        case we choose for each $K$ a random element in each connected
        component of $G$. 
\end{example}
The next example illustrates the technique
	in detail for a given pair of subspaces and illustrates the
	differences
	between the real and complex cases.
\begin{example}
  Let $S= \{0,1,2,3\}$ and let $S' = \{0,1,2,7\}$ be subsets $[0,7]$.
  When $\K = \C$ a random element of $G$ can be taken to have the form
  $g= (e^{\iota \theta_0}, \ldots , e^{\iota \theta_7})$ where
  the $\theta_i$ are drawn randomly from the interval $[0, 2\pi)$.
    In the software Mathematica we used the command
    $G=\text{DiagonalMatrix}\left[e^{i \text{RandomReal}[\{0,2 \pi \},8]}\right]$
    to obtain the element\footnote{We present only the first two significant digits for clear presentation.}:
     $$g  = [-0.26+0.96\I,-0.87-0.47\I,-0.47-0.88\I,-0.33+0.94\I,0.70 +0.71\I,-0.81+0.58\I,-0.59-0.80\I,-0.09+0.99\I]. $$
    The matrix $A_{S,S'}$ is the $8 \times 8$ matrix
  {\scriptsize $$\left( 
    	\begin{array}{cccccccc}
    	1 & 0 & 0 & 0 & 0 & 0 & 0 & 0 \\
    	0 & 1 & 0 & 0 & 0 & 0 & 0 & 0 \\
    	0 & 0 & 1 & 0 & 0 & 0 & 0 & 0 \\
    	0 & 0 & 0 & 1 & 0 & 0 & 0 & 0 \\
    	-0.34+0.25\I & -0.06\I & 0.41+0.26\I & -0.04+0.07\I& 0.18-0.25\I & -0.22+0.16\I & -0.04+0.57\I & -0.21-0.04\I \\
    	-0.21-0.04\I & -0.34+0.25\I & -0.07\I & 0.41 +0.26\I & -0.05+0.08\I & 0.18-0.25\I & -0.22+0.16\I & -0.04+0.57\I \\
    	-0.04+0.57\I & -0.21-0.04\I & -0.34+0.25\I & -0.07\I & 0.41+0.26\I & -0.04+0.07\I & 0.18-0.25\I & -0.22+0.16\I \\
    	-0.07\I & 0.41 +0.26\I & -0.04+0.07\I & 0.18 -0.25\I & -0.22+0.16\I & -0.04+0.57\I & -0.21-0.04\I & -0.34+0.25\I \\
    	\end{array}
    	\right)$$}
which has non-zero determinant and thus maximal rank. It follows
that when $\K = \C$ the general translate of $L_{S}$ does not
intersect $L_{S'}$.

If $\K = \R$ a random element of $G$ has
the form $(\pm 1, e^{\iota \theta_1}, e^{\iota \theta_2}, e^{\iota \theta_3},\pm 1,
e^{-\iota \theta 3}, e^{-\iota \theta 2}, e^{-\iota \theta 1})$.
We take the element
\begin{equation*}
g = [1, 0.44+0.9\I, 0.22+0.97\I,-0.39+0.92\I,1, -0.39-0.92\I,0.22 -0.97\I, 0.44 -0.9\I],
\end{equation*}
and obtain the matrix
{$$\left(
	\begin{array}{cccccccc}
	1 & 0 & 0 & 0 & 0 & 0 & 0 & 0 \\
	0 & 1 & 0 & 0 & 0 & 0 & 0 & 0 \\
	0 & 0 & 1 & 0 & 0 & 0 & 0 & 0 \\
	0 & 0 & 0 & 1 & 0 & 0 & 0 & 0 \\
	0.31 & -0.41 & 0.2 & -0.22 & 0.29 & -0.07 & 0.18 & 0.71 \\
	0.71 & 0.31 & -0.41 & 0.2 & -0.22 & 0.29 & -0.07 & 0.18 \\
	0.18 & 0.71 & 0.31 & -0.41 & 0.2 & -0.22 & 0.29 & -0.07 \\
	-0.41 & 0.2 & -0.22 & 0.29 & -0.07 & 0.18 & 0.71 & 0.31 \\
	\end{array}
	\right)$$}\\
which has rank 7, that is, rank deficient. It follows that $g \cdot L_{S} \cap L_{S'} \neq \{0\}$
and thus we expect that every translate of $L_S$ in the identity component
of $G$ contains sparse vectors. A similar calculation can \rev{be made}  using
  a random element of the other components.

\end{example}

\section{Higher dimensional auto-correlations} \label{sec:higher_dimensions}
Our analysis can also be carried out for higher dimensional periodic auto-correlations.
Here, a signal is function $x \colon [0,N-1]^M \to \K$. We
denote by
$x[\ell_0,\ell_2,\ldots \ell_{M-1}]$
the value of $x$ at $(\ell_0, \ldots , \ell_{M-1}) \in [0,N-1]^M$. 
The periodic auto-correlation function $a_x \colon [0,N-1]^M \to \K$ is given by
$$a_x[n_0, \ldots , n_{M-1}] =  \sum_{(\ell_0, \ldots , \ell_{M-1}) \in [0,N-1]^M} x[\ell_0,
\ldots , \ell_{N-1}] \overline{x[\ell_0 + n_0 , \ldots, \ell_{M-1} + n_{M-1}]},$$
where all indices are considered modulo $N$.
By definition, the periodic auto-correlation obeys \rev{a conjugation-reflection} $\Z_2$ symmetry $a_x[n_0, \ldots , n_{M-1}] = \overline{a_x[N-n_0, \ldots, N-n_{M-1}]}.$

If $K = \C$, then the group $D = (S^1 \times (\Z_N)^M) \ltimes \Z_2$  preserves
the auto-correlation. Here,  
$e^{\iota \phi} \in S^1$ acts by global phase change, $h=(n_0, \ldots, n_{M-1}) \in \Z_N^{M}$
acts by cyclic shift, i.e.,
$$(hx)[\ell_0, \ldots , \ell_{M-1}] = x[\ell_0 + n_0, \ldots , \ell_{M-1} + n_{M-1}],$$  and $(-1) \in \Z_2$ acts
by reflection and conjugation; i.e., $$(-1)\cdot x[\ell_0, \ldots , \ell_{M-1}]
= \overline{x[N-\ell_0, \ldots , N- \ell_{M-1}]}.$$
Similarly, if $\K = \R$
then the group $D = (\pm 1 \times (\Z^N)^M) \ltimes \Z_2$
preserves the periodic auto-correlatoin. In either case we refer to $D$ as the
($M$-dimensional) group of intrinsic symmetries.
Two signals $x \colon [0,N-1]^M \to \K$ are {\em equivalent}
if they are in the same orbit of the group $D$ of intrinsic symmetries. 

Given a subset $S \subset [0,N-1]^M$, we let $L_S$ be
the subspce of signals $[0,N-1]^M \to \K$ whose
support is contained in $S$. 
Let $C$ the set of equivalence classes of $[0,N-1]^M$ modulo
the equivalence relation $(n_0, \ldots , n_{M-1}) \sim (N-n_0, \ldots, N- n_{M-1})$ and let 
$S-S$ to be the cyclic difference set
$\{ (n_0 - m_0, \ldots , n_{M-1}- m_{M-1})\,|\, (n_0, \ldots , n_{M-1}), (m_0,
\ldots m_{M-1}) \} \subset C$. For a generic $x \in L_S$, 
the auto-correlation $a_x$ has $|S-S|$ distinct entries up to reflection and
conjugation. Again for dimension reasons we cannot recover a generic  signal if $|S-S| < |S|$ since
the auto-correlation function, restricted to the subspace $L_S$, can be viewed as a  polynomial function from
$\K^{|S|} \to \K^{|S-S|}$.

As in the one-dimensional case, we expect to be able recover a generic
vector 
in $L_S$ (up to an action of the group $D$ of intrinsic symmetries)
from its higher dimensional auto-correlation $a_x$
provided $|S-S|> |S|$. In other 
words, we expect that the analogue of Conjectures~\ref{conj.main.technical},~\ref{conj.support_recovery}, and \ref{conj.vector_recovery} when $S$ is a subset of $[0, N-1]^M$ with the property that $|S-S| > |S|$ to hold true.
For any specific $S \subset [0,N-1]^M$, this can be verified in a manner
similar to the one-dimensional computational tests, by computing the Hilbert polynomial
of an appropriate incidence variety as in Sections~\ref{subsec.support_recovery},~\ref{subsec.vector_recovery}.

The problem of recovering a signal from its periodic auto-correlation can be extended to signals defined on any finite abelian group $A$ as discussed in Section~\ref{sec:general_group}. 
Under this more general framework, the setups considered in this paper are just special cases: in the one one-dimensional case $A=\Z_N$ and in the multi-dimensional case 
$A=\Z_N^M$.

\appendix


\section{Phase retrieval algorithms and computational complexity}
\label{sec:numerics}

While this work focuses on the question of uniqueness, we would like to
briefly discuss phase retrieval algorithms and the computational complexity of the crystallographic phase retrieval problem;  we refer the reader to~\cite{elser2007searching,elser2017complexity,elser2018benchmark,levin2019note} for further insights. 

\subsection{Phase retrieval algorithms}
Recall that our goal is to find a signal in the intersection of two non-convex sets $x_0\in \Sm\cap\B$~\eqref{eq:problem}.  
We thus define projectors onto these sets; these projectors are simple and can be computed efficiently.  
The projection onto $\B$~\eqref{eq:B} of a general signal $x\in\C^N$ combines the observed Fourier magnitude $y_0$ from~\eqref{eq:model} with the current estimate of the Fourier phase. Formally, the projector onto $\B$ is defined by
\begin{equation}
P_\B(x) = F^{-1} (y_0\odot \sign(Fx)),
\end{equation}
where $'\odot'$ denotes an element-wise product and $\sign(x)[n]=\frac{x[n]}{|x[n]|}$  for any $x[n]\neq 0$ and  $\sign(x)[n]=0$ otherwise. The projector onto $\Sm$ leaves the $K$ entries with the largest absolute values intact, and zeros out all other entries. Therefore, $P_\Sm(x)$ is a K-sparse signal by definition. 

A naive approach to solve the X-ray crystallography phase retrieval problem, and phase retrieval in general, is to apply the two projectors iteratively, i.e.,
\begin{equation}
x \mapsto P_\Sm P_\B (x).
\end{equation}
This scheme is called alternating projection in the mathematics literature, and Gerchberg-Saxton in the phase retrieval literature.
Unfortunately, for hard problems such as  crystallographic phase retrieval, this scheme tends to stagnate quickly in points far away from a solution. 

Alternatively, algorithmic schemes which are close relatives of the Douglas-Rachford splitting algorithm~\cite{douglas1956numerical,lindstrom2018survey,li2016douglas} \rev{and the alternating direction method of multipliers (ADMM)} have been proven to be highly effective. 
These algorithms are based on the reflection operators, defined as $R_\B = 2P_\B-I$ and $R_\Sm = 2P_\Sm-I$, where $I$ is the identity operator.
One representative, simple yet effective, algorithm is called \emph{relaxed reflect reflect} (RRR). For a fixed parameter $\beta\in(0,2)$, the RRR iterations read:
\begin{equation}
x \mapsto x + \frac{1}{2}(I + \beta R_\B R_\Sm)\rev{(x)}, 
\end{equation} 
or, more explicitly, 
\begin{equation}
x \mapsto x + \beta (P_\B(2P_\Sm(x)-x ) - P_\Sm(x)).
\end{equation} 
For $\beta=1$ this algorithm coincides with Douglas-Rachford. Other variations of Douglas-Rachford that are used in practice include 
Fienup's hybrid input-output (HIO) algorithm~\cite{Fienup1982}, 
the difference map algorithm~\cite{elser2003phase}, and the relaxed averaged alternating reflections (RAAR) algorithm~\cite{Luke2005}.
In addition to phase retrieval, these algorithms seem to be surprisingly effective  for a variety of challenging feasibility problems, such as the Diophantine equations, sudoku, and protein conformation determination~\cite{elser2007searching}; recently, it was even applied to deep learning~\cite{elser2019learning}. 
One specific interesting property of RRR (and most of its relatives) is that|in contrast to optimization-based algorithms|it stagnates only when it finds a point \rev{from which the intersection  $\Sm\cap\B$ can be found trivially by projection}.  Note that this property does not guarantee finding a solution in a finite number of steps.

\subsection{Computational complexity}
Strong empirical evidence suggest that the computational complexity of  RRR for the crystallographic phase retrieval
	 problem increases exponentially fast with $K$~\cite{elser2018benchmark}, however, rigorous theoretical analysis is lacking.

To illustrate the computational complexity, we ran RRR with step size parameter $\beta=1/2$ (chosen empirically), $N=50$ and varying $K$, and counted how many iterations are required to reach a solution from a random initialization\footnote{The code to reproduce this experiment, as well as to re-generate all other figures in the paper, is publicly available at \url{https://github.com/TamirBendory/crystallographicPR}.}.
To measure the error while taking symmetries into account, we define
\begin{equation} \label{eq:error}
\text{error} = \min_{g\in D}\frac{\|g\cdot x_{\text{est}}-x_0\|_2^2}{\|x_0\|_2^2},
\end{equation}
where $x_{\text{est}}$ is the estimated signal and $D=\Z_2\times D_{2N}$|the group of intrinsic symmetries. 
A solution was declared when the error dropped below $10^{-8}$.
Plainly, this measure cannot be used in practice since it requires knowing the sought signal, but it suffices for the purposes of this work. 
In practice, a natural error measure is $\eta = \frac{\|P_\Sm(x)\|_2^2}{\|x\|_2}$: 
this index measures the portion of the signal's energy concentrated in the dominant $K$  entries of the current estimate~\cite{elser2018benchmark}. 
To generate the underlying signal, we drew a random  set of $K$ indices from $[0,N-1]$ to form the support set $S$.
Then, each entry $x[i]$ for $i\in S$ was drawn i.i.d.\ from a uniform distribution over $[0,1]$. The rest of the entries were set to zero. 
Figure~\ref{fig:iter_count} shows that the median number of iterations required to reach a solution grows exponentially fast.
We believe that this is not a flaw of RRR, but an indication for the  computational hardness of the crystallographic phase retrieval problem, regardless of any specific algorithm. In particular, as far as we know, there are no  polynomial-time algorithms for this problem.
The iteration counts also
display a considerable variability. 

\begin{figure}[ht]
	\begin{subfigure}[h]{0.5\columnwidth}
	\centering
	\includegraphics[width=\columnwidth]{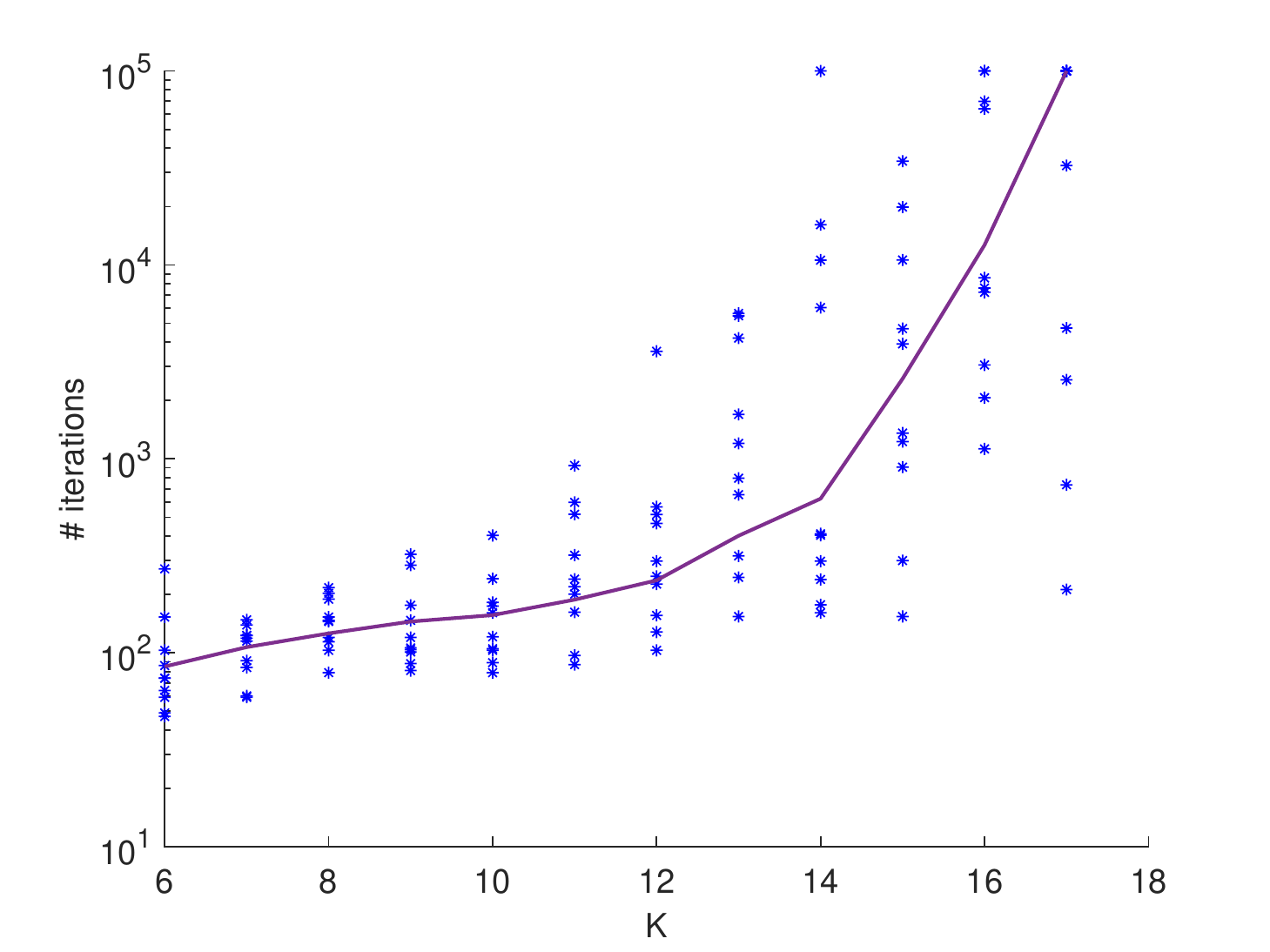}
	\caption{N = 50}
\end{subfigure}
\hfill
	\begin{subfigure}[h]{0.5\columnwidth}
	\centering
	\includegraphics[width=\columnwidth]{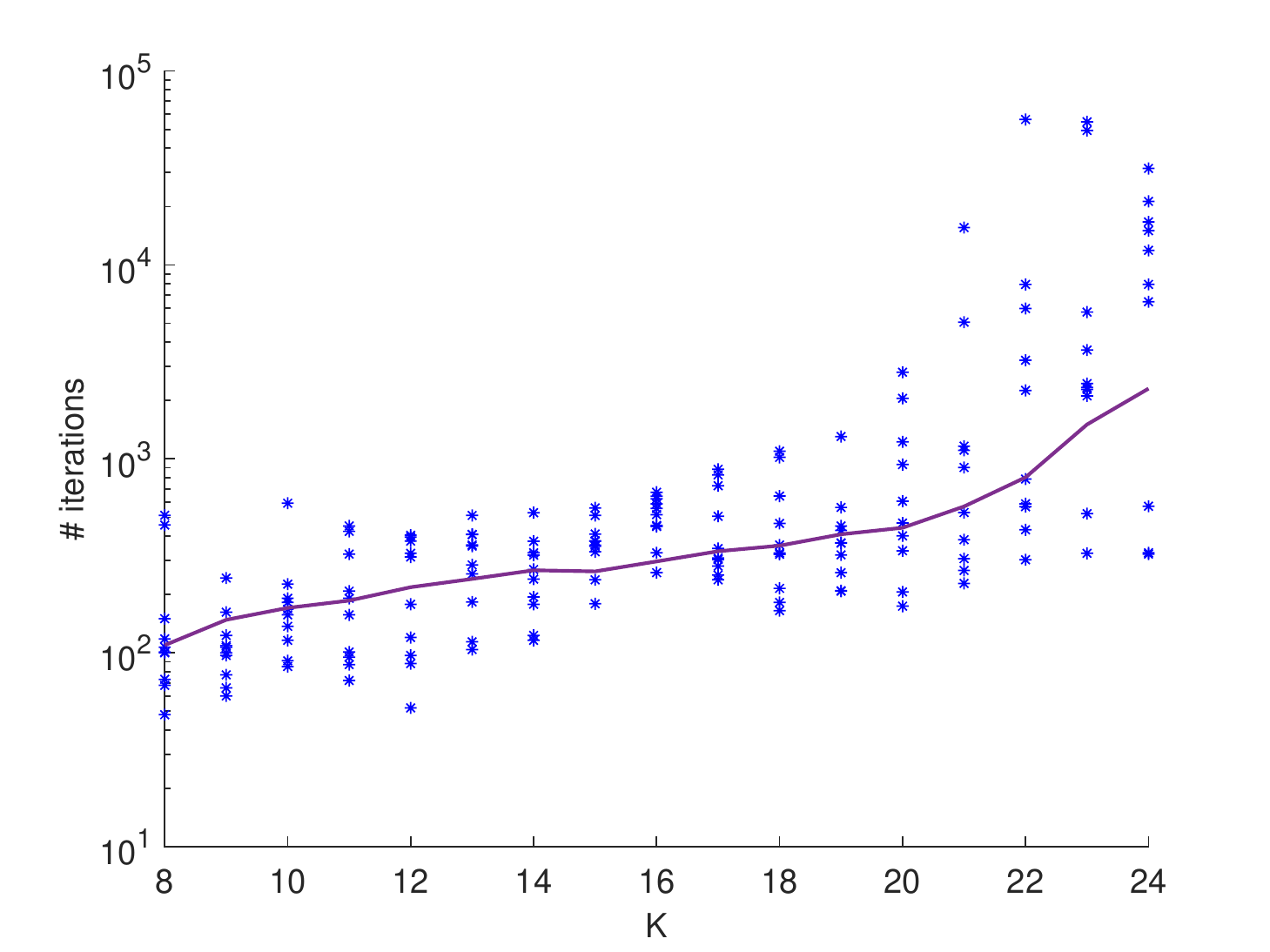}
	\caption{N = 100}
\end{subfigure}
\hfill
\caption{\label{fig:iter_count} The median number of RRR
 iteration counts (running time) over 500 trials per $K$ for $N=50$ (left) and $N=100$ (right). As can be seen, the iteration counts  grow exponentially fast with $K$.
		The blue asterisks present the specific iteration count of 10 individual trials per~$K$, and are used only to illustrate the high variability of the results. }  
\end{figure}

The exponential computational complexity of RRR  restricts our ability to empirically verify the conjectured uniqueness limit  $K\approx N/2$ for large  values of  $N$. \rev{Unfortunately, for small $N$, there are few subsets $S$ that satisfy the necessary condition $|S-S|>K$.} 
As a compromise, we conducted an experiment with $N=8$ and $K=3,4$. For each $K$, we ran 1000 trials with random support sets that satisfy $|S-S|>K$.
The maximum number of iterations was set to $10^7$. 
 For $K=3$, 77 trials out of 1000 reached the maximal number of iterations. In other words, $85\%$ of the trials were declared successful. For $K=4$, only $60\%$ of the trials were successful. Figure~\ref{fig:iter_hist} shows the empirical distribution of the iteration counts, which decays in exponential rate.   

\begin{figure}[ht]
	\begin{subfigure}[h]{0.5\columnwidth}
		\centering
		\includegraphics[width=\columnwidth]{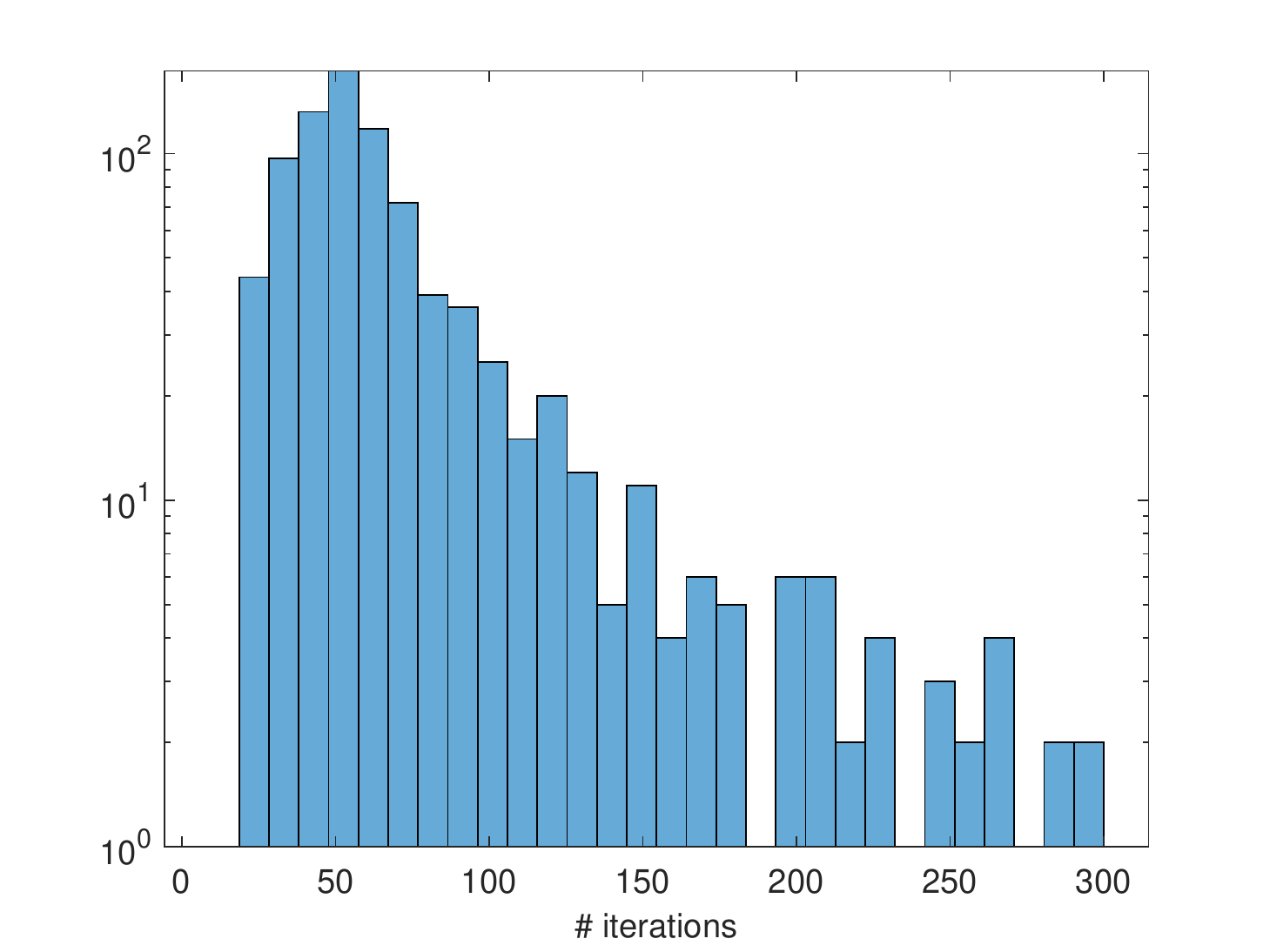}
		\caption{K=3}
	\end{subfigure}
	\hfill
	\begin{subfigure}[h]{0.5\columnwidth}
		\centering
		\includegraphics[width=\columnwidth]{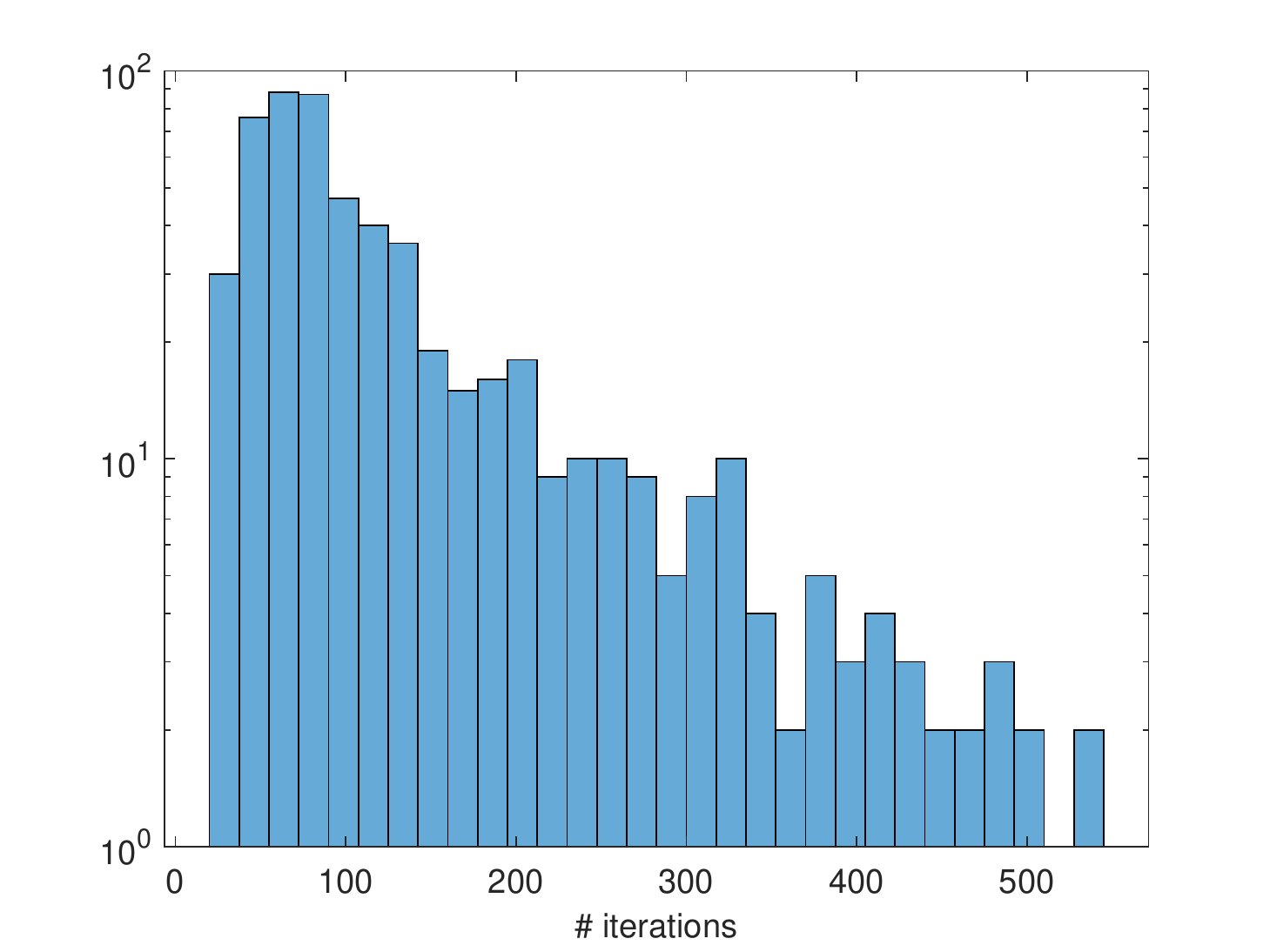}
		\caption{K=4}
	\end{subfigure}
	\caption{\label{fig:iter_hist} Iteration counts (running time) histograms  over 1000 trials for $N=8$ and $K=3,4$. The iteration counts decay in exponential rate. Only the left ends of the histograms (that include almost all trials) are presented for clear visualization.}  
\end{figure}

\section{Density of sets with small difference sets} \label{app.difference}
For a given value of $K$ and $N$, an important mathematical question
is to estimate the number of $K$-element subsets $S$ with the property
that $|S-S| \leq K$. 
This question is quite subtle and relates to some deep problems in additive number theory. It is beyond the scope of this paper to obtain this analysis, but classic results of Kemperman~\cite{kemperman1960sumset} (see also~\cite{vsevolod1999sumset}) give a technique for enumerating the
sets $S$ with this property. This classification is somewhat involved and depends on the prime factorization of $N$. However, if $N$ is prime then
Kemperman's results imply the following.
\begin{prop} \label{prop.kemperman}
	If $N$ is prime,  then $|S-S| \leq |S|$ if and only if $S$ is an arithmetic progression.
\end{prop}
\begin{proof}
	Denote by $S-^{\Z_N} S \subset \Z_N$ the set of
	differences $\{i-j| i , j \in S\} \subset \Z_N$. (Here we do not identify
	an element and its negative in $\Z_N$.) Because $0 \in S-^{\Z_N} S$ is its own negative, it follows that 
	$|S-^{\Z_N} S| \leq 2|S-S|-1$. Hence, if $|S-S| \leq  |S|$
	then $|S-^{\Z_N}S| < 2|S| = |S| + |-S|$ where $-S = \{-i| i \in S\}$.
	In this case~\cite[Corollary, p.~74]{kemperman1960sumset}  implies that $S$
	is an arithmetic progression.
	
	Conversely, if $S = \{a_0, a_0 + d, \ldots , a_0 + (d-1)K\}$
	then $S-S = \{0, \overline{d}, \ldots , \overline{(d-1)K}\}$, where $\overline{m}$ indicates the element of $[0,  N/2 ]$
	corresponding to the equivalence class of $m$ under the equivalence
	relation $m \sim -m$.
\end{proof}
Let $\S_K$ be the set of $K$-element subsets of $[0,N-1]$ and
let $\T_K$ be the set of $K$-element subsets of $[0,N-1]$ such
that $|S-S| < |S|$.
The following corollary says that, at least when $N$ is prime, the probability
of picking a subset with $|S-S| < |S|$ drops quickly to 0 as $N \to \infty$.
\begin{prop} \label{pro:prime}
	For prime $N$ and $K/N \leq 1/2$,
	the ratio $|\T_K|/|\S_K|$ tends to 0 as $N \to \infty$.
\end{prop}
\begin{proof}
	By Proposition~\ref{prop.kemperman}, when $N$ is prime
	$S \in \T_K$ if and only if $S$ is an arithmetic progression of length $K$.
	The equivalence class of an arithmetic progression is determined
	by its difference~$d \in \Z_N$. Moreover, any progression with   difference $d$ is equivalent under the action of the dihedral group
	to a progression with difference $N-d$. Thus, the number
	of equivalence classes of arithmetic progressions equals $ N/2 $. Since the dihedral group $D_{2N}$ has $2N$ elements we see that
	that the total number of arithmetic progressions is $\sim N^2$.
	On the other hand, the total number of $K$-element subsets is
	${N \choose K}$. Thus, the ratio
	$|\T_K|/|\S_K| \sim {N^2\over{{N \choose K}}}$ which goes quickly to 0
	as $N \to \infty$.
\end{proof}

We expect that a more refined analysis using Kemperman's classification will
show that the number of equivalence classes of $S$ such that $|S-S| \leq K$
is asymptotic to 0 even for composite $N$; see Figures~\ref{fig:hist} for supporting empirical evidences.
However, deriving such a results analytically is beyond the scope of the current work. 

\section{Proof of Proposition~\ref{prop:collisions}}
\label{sec:proof_prop_collisions}
	In order for $S$ to be collision-free we need that every non-zero element
	of the multi-set $S-S$ appears with multiplicity exactly one so
	that $|S-S|$ is maximized. If $|S| = K$
	then the number of  non-zero differences (counted with multiplicity)
	in $S-S$ is ${K \choose 2}$. Thus, $S$ is collision free if and only
	if $|S-S| = {K \choose 2} +1$. (We add one because $0 \in |S-S|$).
	Thus, a necessary condition of $[0,N-1]$ to contain any collision-free subset
	is that ${K \choose 2} +1  \leq N$. For any fixed value of $R$,
	the function ${RN \choose 2} + 1$ grows quadratically in $N$. Therefore,	for $N$ sufficiently large ${K \choose 2} + 1 > N$ so there
	can be no collision-free subsets.

%


\section{Hilbert polynomial, dimensions and degrees of varieties}
\label{sec:hilbert}
Consider the polynomial ring $R=\K[x_0, \ldots , x_n]$
where, $\K$ is a field.
For each $d$, the set $R_d$ consisting of homogeneous polynomials
of degree $d$ is a finite dimensional $\K$-vector subspace with basis consisting of
the monomials of degree $d$ in $x_0, \ldots , x_n$. A well known
combinatorial formula for the number of monomials implies that 
$\dim_\K R_d = {{n+d}\choose{d}}$. For example, if $n =1$
then $\dim R_d$ is the number of binary forms of degree 2 in $(n+1)$-variables
which is $d+1 = {{d+1}\choose{d}}$. Note that function
$d \mapsto \dim R_d$ is a polynomial in $d$ of degree $n$.

Given a set of homogenous polynomials $f_1, \ldots , f_r$, 
let $I=(f_1, \ldots , f_r)$ be the ideal they generate. The {\em Hilbert function} $H_I$ is defined as the function
$d \mapsto \dim(R/I)_d$, where $(R/I)_d$ denotes the subspace
of $R/I$, consisting of homogeneous elements of degree $d$.
The {\em Hilbert-Serre Theorem} \cite[Theorem I.7.5]{hartshorne1977algebraic_geometry} states that there exists an integer
valued polynomial
$P_I$ such that for $d\gg 0$,
$H_I(d) = P_I(d)$. The polynomial $P_I$ is called the {\em Hilbert polynomial}
of $I$. If we set $P_k$ to be the polynomial $P_k(d) = {{k+d} \choose {d}}$
then we can write
$$P_I = a_\ell P_l + a_{\ell -1} P_{l-1} + \ldots a_0 P_0$$
with $a_0, \ldots , a_{\ell}$ integers and $a_{\ell} > 0$.

In addition, the Hilbert-Serre theorem implies that 
$\deg P_I$ equals the dimension of the subvariety of
the projective space $\Pro^n$ defined by the homogeneous polynomials
$f_1, \ldots , f_r$.
Equivalently if we consider $Z(f_1, \ldots , f_r)$ as a subset
of $\K^n$, then $\deg P_I = \dim Z(I) -1$.
Moreover, the coefficient $a_{\ell}$ is positive and equals to the degree of
$Z(I)$ as a projective variety,
where the degree of a projective variety $Z(I)$ of dimension $\ell$
is defined as the number of points
in the intersection $Z(I) \cap L_{n-\ell}$, where $L_{n -\ell}$ is a general
linear subspace of dimension $n-\ell$~\cite[Theorem 7.7]{hartshorne1977algebraic_geometry}.

Using the Hilbert function we can obtain the following proposition
which we use in Section \ref{subsec.vector_recovery}.
\begin{prop} \label{prop.linear_strip}
	Suppose that $\dim Z(I) = \ell$ and has degreee $a$. If $Z(I)$
	contains $a$  $\ell$-dimensional linear subspaces $L_1, \ldots , L_a$ then
	$\dim Z(I) \setminus (L_1 \cup \ldots \cup L_a) < \ell$.
\end{prop}
\begin{proof}
	Let $I_L$ be the ideal generated by the linear forms defining
	the subspaces $L_1, \ldots , L_a$. By~\cite[Proposition 7.6]{hartshorne1977algebraic_geometry},  $Z(I_L)$ has dimension $\ell$ and degree $a_\ell$.
	Thus, $P_{I_L} = a_{\ell} P_\ell + \tilde{P}$ for some lower degree terms $\tilde{P}$.
	Hence, $Z(I)$ and $L_1 \cup \ldots \cup L_{a}$ have the same degree and dimension.
	Let $Y$ be the closure of  $Z(I) \smallsetminus (L_1 \cup \ldots \cup L_a)$
	in $\Pro^n$.
	Then $Z(I) = Y \cup (L_1 \cup \ldots \cup L_a)$. Since
	$Y \subset Z(I)$ we know that $\dim Y \leq \dim Z(I)$. Suppose
	that $\dim Y = \dim Z(I)$. Then by~\cite[Proposition 7.6b]{hartshorne1977algebraic_geometry}, $\deg Z(I) = \deg Y + \deg (L_1 \cup \ldots \cup L_a)$.
	But since $\deg Z(I) = a$ this a contradiction. Hence, $\dim Y < \dim Z(I)$.
\end{proof}

The Hilbert function of an ideal generated by polynomials $f_1, \ldots
, f_r$ with rational coefficients can be be computed exactly using a
computer algebra system. This is automated in two steps, which are executed by the command ${\tt hilbertPolynomial}$ in 
Macaulay2~\cite{M2}.
The first
step is to replace the generators $f_1, \ldots , f_r$ of the ideal $I$
with new generators $g_1, \ldots , g_t$ called a Gr\"obner
basis; see~\cite[Chapter 15]{eisenbud1995commutative_algebra} for
the definition of a Gr\"obner basis.
Given a Gr\"obner basis, the problem of computing the Hilbert
polynomial of an ideal is combinatorial.  Both steps can be computed
to infinite precision using a computer algebra system. 
Although neither step can be performed in polynomial time, implemented
algorithms are efficient when the number of variables is relatively
small. 

\section{Sparse periodic phase retrieval in finite abelian groups} \label{sec:general_group}
The sparse phase retrieval problem can be generalized to any finite abelian group. Let $A$ be a finite abelian group. We denote the composition operation by $+$,
the identity by $0$. and the inverse of an element $a$ as $-a$.
Let $V$ be the $\K$-vector space of functions $x \colon A \to \K$.
In the case of one-dimensional phase retrieval $A = \Z_N$ is a cyclic group, and 
in the case of higher dimensional phase retrieval $A = \Z_N^M$ is
a product of cyclic groups of the same order.
The auto-correlation of $x \in V$
is the function $a_x \colon A \to \K$ defined
by the formula 
\begin{equation}
a_x[\ell] =\sum_{\ell' \in A} x[\ell']\overline{x[\ell+\ell']}.
\end{equation}
The function $a \colon V \to V, x \mapsto a_x$
is invariant under the group $D_A = (S^1 \times A) \ltimes \Z_2$ if $\K = \C$
or $D_A = (\pm 1 \times A) \ltimes \Z_2$ and $\K = \R$.
Here,  $S^1$ (resp. $\pm 1$) acts by a scalar multiplication,
$A$ acts by translation, that is, $$(\ell \cdot x)[\ell'] = x[\ell'+\ell],$$ for some $\ell\in A$,  and
$\Z_2$ acts by conjugation and reflection, i.e., 
\begin{equation*}
(-1\cdot x)[\ell] = \overline{x[\ell]}.
\end{equation*}

If we let $C = A/\Z_2$, where $\Z_2$ acts on $A$ by $(-1\cdot \ell) = -\ell$,
then we can define the ``difference set'': $S-S = \{ \ell_1 - \ell_2 \,|\, \ell_1, \ell_2 \in A\} \subset C$.
With this setup, our main conjecture is as follows.
\begin{conj} \label{conj.technical.generalgroup}
	Suppose that $S$ is a subset of $A$ such that $|S-S| > |S|$ and $x \in L_S$ is a generic signal. Then,
	$a_x = a_{x'}$ implies that
	$x'$ is obtained from $x$ by an action of the group~$D_A$ of  intrinsic symmetries. 
\end{conj}

Similarly, we can formulate  general group-theoretic versions of Conjectures
~\ref{conj.support_recovery} and~\ref{conj.vector_recovery}.
To establish notation we note that the group $A \ltimes \Z_2$ (the analog of the dihedral group)
acts on the set of subsets of $A$, where $a \in A$ acts by
translation; i.e., $a + S = \{ a+ s | s\in S\}$ 
and the non-trivial element in $\Z_2$ acts by ``reflection,''
i.e., it maps $S$ to $-S=\{-s| s\in S\}$. 
We say that two subsets $S,S'$ are equivalent if $S' = g \cdot S$ for
some $g \in A \ltimes \Z_2$. Given a subset $S \subset A$,
we denote by $D_{S,A}$ the subgroup of $D_A$
that preserves $L_S$ and again refer to it as the group of intrinsic symmetries
of the subspace $L_S$.

\rev{
\begin{conj} \label{conj.support_recovery_general_group}
	Suppose that $S$ and $S'$ are two non-equivalent $K$-element subsets of
	an abelian group $A$
	with $|S-S| =|S' -S'| > K$. Then, for generic $x \in L_S$, $a(x)$ is not in $a(L_{S'})$. Namely, the support of $x$ is determined up to equivalence
        under the action of the group $A \ltimes \Z_2$ by the periodic auto-correlation of $x$.
\end{conj}
}
\begin{conj} \label{conj.vector_recovery_general_group}
	Suppose that $|S - S| > |S| $. 
	If $x \in L_S$ is a generic
	vector and $x' \in L_S$ is another vector (in the same subspace) such that $a(x) = a(x')$,
	then $x' = g \cdot x$ for some~$g \in D_{S,A}$.
\end{conj}

\section*{Acknowledgments}
The authors thank Ti-Yen Lan for his comments on an early draft of this paper, \rev{and the anonymous referees for their valuable comments and suggestions.}

\bibliographystyle{siamplain}

\end{document}